\documentclass[a4paper,10pt]{article}

\usepackage[utf8]{inputenc}
\usepackage[T1]{fontenc}
\usepackage{indentfirst}

\usepackage{geometry}
\usepackage{color}
\usepackage{verbatim}
\usepackage{amstext}
\usepackage{amsthm}
\usepackage{amssymb}
\usepackage{graphicx}
\usepackage{float}
\usepackage[mathscr]{euscript}

\usepackage{amsmath,amssymb,amsthm,bbm}
\usepackage{enumerate} 
\usepackage[toc,page]{appendix}
\usepackage{authblk}
\usepackage[svgnames]{xcolor}
\usepackage[pdftex]{hyperref}
\hypersetup{colorlinks=true}
\hypersetup{%
    colorlinks,
    linkcolor=Navy,
    citecolor=Navy,
    urlcolor=Navy
}
\usepackage{graphicx}
\usepackage{epstopdf}
\usepackage{booktabs}
\usepackage{multirow}
\hypersetup{colorlinks=true}
\hypersetup{
    colorlinks,
    linkcolor=Navy,
    citecolor=Navy,
    urlcolor=Navy
}

\usepackage{titlesec}

\titleformat{\paragraph}
{\normalfont\normalsize\bfseries}{\theparagraph}{1em}{}
\titlespacing*{\paragraph}
{0pt}{3.25ex plus 1ex minus .2ex}{1.5ex plus .2ex}

\newtheorem{theorem}{Theorem}[section]

\newtheorem{cor}[theorem]{Corollary} 
 
\newtheorem{defn}[theorem]{Definition}

\newtheorem{lem}[theorem]{Lemma}
\newtheorem{thm}[theorem]{Theorem}

\allowdisplaybreaks

\title{Variance and interest rate risk \\in unit-linked insurance policies}
\author[1]{David Baños}
\author[1]{Marc Lagunas-Merino}
\author[1]{Salvador Ortiz-Latorre}
\affil[1]{Department of Mathematics, University of Oslo.}

\date{\today}
\begin{document}

\maketitle

\begin{abstract}
One of the risks derived from selling long term policies that any
insurance company has, arises from interest rates. In this paper we
consider a general class of stochastic volatility models written in
forward variance form. We also deal with stochastic interest rates to
obtain the risk-free price for unit-linked life insurance contracts,
as well as providing a perfect hedging strategy by completing the market. We conclude
with a simulation experiment, where we price unit-linked policies using Norwegian mortality rates. In addition we compare prices for the classical Black-Scholes model against the Heston stochastic volatility model with a Vasicek interest rate model. 
\end{abstract}




\section{Introduction}

A unit-linked insurance policy is a product offered by insurance companies. Such contract specifies an event under which the insured of the contract obtains a fixed amount. Typically, the payoff of such contract is the maximum value between some prescribed quantity, the guarantee, and some quantity depending on the performance of a stock or fund. For instance, if $G$ is some positive constant amount, and $S$ is the value of some equity or stock at the time of expiration of the contract, then a unit-linked contract pays
$$H=\max \{G, f(S)\},$$
where $f$ is some suitable function of $S$. Here, the payoff $H$ is always larger than $G$, hence being $G$ a minimum guaranteed amount that the insured will receive. Naturally, the price of such contract depends on the age of the insured at the moment of entering the contract and the time of expiration, likewise, it also depends on the event that the insured is alive at the time of expiration.

The risk of such contracts depends on the risk of the financial instruments used to hedge the claim $H$, and there are many ways to model it. The most classical one is considering the evolution of $S$ under a Black-Scholes model, this is for instance the case in \cite{BoyleSchwartz77} or \cite{AasePersson2011}, where the authors derive pricing and reserving formulas for unit-linked contracts in such setting. One can also consider a more general class of models. For example, it is empirically known that the driving volatility of $S$ is, in general, not constant. One could then take a stochastic model for the volatility, as it is done in \cite{WaQiWa13}, where the authors carry out pricing and hedging under stochastic volatility. Since there is more randomness in the model, complete hedging is no longer possible, the authors in \cite{WaQiWa13} provide the so-called local risk minimizing strategies.

In this paper instead, we look at the problem from two different perspectives. On the one hand, we also consider stochastic volatility, as market evidence shows. Nonetheless, there are available instruments in the market for hedging against volatility risk, the so-called forward variance swaps. Such products are contracts on the future performance of the volatility of the stock. In such a way, we want to price unit-linked contracts taking into account that the insurance company can trade these instruments, as well. On the other hand, it is known that unit-linked products share similarities with European call options. For example, authors in \cite{BoyleSchwartz77} recognize the payoff of unit-linked products as European call options plus some certain amount. However, European call options have very short maturities, typically between the same day of the contract up to two years, while it is not uncommon to have unit-linked insurance contracts that last for up to 40 years. For this reason, there is an inherent risk in the interest rate driving the intrinsic value of money. In this paper we take such long-term risk into account, as well. Classically, most of the literature about equity-linked policies assume deterministic interest rates. Nevertheless, some research on stochastic interest rates has also been carried. For example, in \cite{BP02} the authors consider stochastic interest rates under the Heath-Jarrow-Morton framework and study different types of premium payments. In addition, a comparison with the classical Black-Scholes model is offered in \cite{BP02}. Also in \cite{BO93}, the Vasicek and Cox-Ingersoll-Ross model is considered for the interest rate. In this paper we consider a general framework including both cases. 

While many results in the literature deal with the construction of risk minimizing strategies in incomplete markets, in this paper instead, inspired by \cite{RomanoTouzi97}, we complete the market by allowing for the possibility to trade other instruments that one can find in the market. On the one hand, we introduce zero-coupon bonds to hedge against interest rate risk and, on the other hand, we introduce variance swaps to hedge against volatility risk.

This paper is organized as follows. First, we introduce in Section \ref{framework} our insurance and economic framework with the specific models for the money account, stock and volatility. Then, in Section \ref{unit-linked Section}, we complete the market by incorporating zero-coupon bonds and variance swaps in the market. We derive the dynamics of the new instruments used to hedge and apply the risk-neutral theory to price insurance contracts subject to the performance of an equity or fund with stochastic interest and volatility. In Section \ref{sec: 4} we take a particular model; the Vasicek model for the interest rate and a Heston model written in forward variance form. We implement the model and do a comparison study with the classical Black-Scholes model in Section \ref{implementation}, where we generate price surfaces under Norwegian mortality rates and different maturities. We conclude Section \ref{implementation} with a Monte-Carlo simulation of the price distributions.


\section{Framework}\label{framework}

The two basic elements needed in order to build a financial model
robust enough to be able to price unit-linked policies, are a financial
market and a group of individuals to write insurance on. We consider a finite time horizon $T>0$ and a given probability space $\left(\Omega,\mathcal{A},\mathbb{P}\right)$ where $\Omega$ is the set of all possible states of the world, $\mathcal{A}$ be a $\sigma$-algebra of subsets of $\Omega$ and $\mathbb{P}$ be a probability measure on $(\Omega,\mathcal{A})$. We model the information flow at each given time with a filtration $\mathcal{F}=\left\{ \mathcal{F}_{t},t\in\left[0,T\right]\right\} $ given by a collection of increasing $\sigma$-algebras, i.e. $\mathcal{F}_{s}\subset\mathcal{F}_{t}\subset\mathcal{A}$ for $t\geq s$. We will also assume that $\mathcal{F}_{0}$ contains all the sets of probability zero and that the filtration is right continuous. We also take $\mathcal{A}=\mathcal{F}_{T}$. The information flow $\mathcal{F}$ comes from two sources; the financial market and the states of the insured that are relevant in the policy. The market information available at time $t$ will be denoted by $\mathcal{G}_t$ and the information regarding to the state of the insured by $\mathcal{H}_t$. We will assume throughout the paper that the $\sigma$-algebras $\mathcal{G}_{t}$
and $\mathcal{H}_{t}$ are independent for all $t$, which implies that the value
of the market assets is independent of the health condition of the insured. We also assume that $\mathcal{F}_{t}=\mathcal{G}_{t}\vee\mathcal{H}_{t}$, for all $t$, where $\mathcal{G}_{t}\vee\mathcal{H}_{t}$ is the $\sigma$-algebra
generated by the union of $\mathcal{G}_{t}$ and $\mathcal{H}_{t}$.
This can be understood as the total amount of information available
in the economy at time $t$, that is the information one can get by recording
the values of market assets and the health state of the insured from
time 0 to time $t$.



\subsection{The market model}

The market information $\mathcal{G}$ will be modeled using the filtration generated by three independent standard Brownian motions, $W_{t}^{0},W_{t}^{1}$ and $W_{t}^{2}$.
These three Brownian motions represent the sources of risk in our model.
We will consider a market formed by assets of two different natures.
A bank account, considered of riskless nature and stock or bond prices, which are of risky nature.

We start by defining the bank account, whose price process is denoted
by $B=\left\{ B_{t}\right\} _{t\in\left[0,T\right]}$, such that $B_{0}=1$.
We will assume the asset evolves according to the following differential equation
\begin{equation}
dB_{t}=r_{t}B_{t}dt,\quad t\in \left[0,T\right],\label{eq: Money_Market_Account_SDE}
\end{equation}
where $r_{t}$ is the instantaneous spot rate and it is assumed to
have integrable trajectories. Actually, we will assume that this rate
evolves under the physical measure $\mathbb{P}$, according to the
following SDE
\begin{equation}
dr_{t}=\mu\left(t,r_{t}\right)dt+\sigma\left(t,r_{t}\right)dW_{t}^{0},\quad r_0 >0, \quad t\in\left[0,T\right],\label{eq: short_rate_SDE}
\end{equation}
where $\mu,\sigma:\left[0,T\right]\times\mathbb{R}\rightarrow\mathbb{R}$
are Borel measurable functions such that, for every $t\in\left[0,T\right]$
and $x\in\mathbb{R}$
\[
\left|\mu\left(t,x\right)\right|+\left|\sigma\left(t,x\right)\right|\leq C\left(1+\left|x\right|\right),
\]
for some positive constant $C$, and such that for every $t\in\left[0,T\right]$
and $x,y\in\mathbb{R}$
\[
\left|\mu\left(t,x\right)-\mu\left(t,y\right)\right|+\left|\sigma\left(t,x\right)-\sigma\left(t,x\right)\right|\leq L\left|x-y\right|,
\]
for some constant $L>0$. We will also assume there exists $\epsilon>0$,
such that $\sigma\left(t,x\right)\geq\epsilon>0$ for every $\left(t,x\right)\in\mathbb{R}_{+}\times\mathbb{R}$.
These conditions are sufficient to guarantee a unique global strong
solution of $\left(\ref{eq: short_rate_SDE}\right)$, weaker conditions
may be imposed, see e.g. \cite[Chapter IX, Theorem 3.5]{Revuz99}.

One of the risky assets will be the stock. We describe the stock price process $S=\left\{ S_{t}\right\} _{t\in\left[0,T\right]}$
by a general mean-reverting stochastic volatility model. Specifically, we will consider the following SDEs
\begin{align}
\frac{dS_{t}}{S_{t}} & =b\left(t,S_{t}\right)dt+a\left(t,S_{t}\right)f\left(\nu_{t}\right)dW_{t}^{1}, \quad S_0 >0,\label{eq: Stock_Price_SDE}\\
d\nu_{t} & =-\kappa\left(\nu_{t}-\bar{\nu}\right)dt+ h\left(\nu_{t}\right) dW_{t}^{2},\quad \nu_0>0, \label{eq: Instantaneous_Variance}
\end{align}
for $t\in\left[0,T\right]$. Here $a,b$ are uniformly
Lipschitz continuous and bounded functions, such that $a\left(t,x\right)>0$ for
all $\left(t,x\right)\in\left[0,T\right]\times\mathbb{R}$. The function $f$ is assumed to be continuous with linear growth and strictly positive. 
We assume that $h$ is a non-negative, linear growth, invertible function such that,
\[\lvert h(x)-h(y)\rvert^2\leq\ell\left(\lvert x-y\rvert\right),\]
for some function $\ell$ defined on $\left(0,\infty\right)$ such that
\[\int_0^\epsilon\frac{dz}{\ell\left(z\right)}=\infty,\quad \text{for any } \epsilon >0.\]
Then, \cite[Chapter IX, Theorem 3.5(ii)]{Revuz99} guarantees the existence of a pathwise unique solution of equation $\left(\ref{eq: Stock_Price_SDE}\right)$. We call $\nu_{t}$, the instantaneous variance. 

Due to the fact that neither $\nu$ nor $r$ are tradable, our market model is highly incomplete. In the forthcoming section, we will complete the market by introducing financial instruments in order to hedge against the risk derived from instantaneous variance and interest rates.

We introduce the num\'eraire, with respect to which we will discount our cashflows

\begin{defn}
\label{def: Stochastic Discount Factor} The (stochastic) discount
factor $D_{t,T}$ between two time intervals $t$ and $T$, $0\leq t\leq T,$
is the amount at time $t$ that is equivalent to one unit of currency
payable at time $T$, and is given by
\begin{equation}
D_{t,T}=\frac{B_{t}}{B_{T}}=\exp\left(-\int_{t}^{T}r_{s}ds\right).\label{eq: Stochastic_Discount_Factor}
\end{equation}
\end{defn}

\subsection{The insurance model}

In what follows, we introduce our insurance model. More specifically, we want to model the insurance information $\mathcal{H}$ as the one generated by a regular Markov chain $X=\{X_t, t\in [0,T]\}$ with finite state space $\mathscr{S}$ which regulates the states of the insured at each time $t\in [0,T]$. For instance, in an endowment insurance, the state $\mathscr{S}=\{\ast,\dag\}$ consists of the two states, $\ast=$"alive" and $\dag=$"deceased". In a disability insurance we have three states, $\mathscr{S}=\{\ast,\diamond,\dag\}$ where $\diamond$ stands for "disabled". Observe that $X$ is right-continuous with left limits and, in particular, $\mathcal{H}$ satisfies the usual conditions.

Introduce the following processes:
$$I_i^X(t) = \begin{cases} 1, \mbox{ if } X_t=i,\\ 0, \mbox{ if } X_t\neq i \end{cases},\quad i\in \mathscr{S},$$
$$N_{ij}^X(t) = \# \{s\in (0,t): X_{t^-} = i, X_t=j \},\quad i,j\in \mathscr{S}, \quad i\neq j.$$
Here, $\#$ denotes the counting measure and $X_{t^-}\triangleq \lim_{\substack{u\to t\\ u<t}}X_u$ the left limit of $X$ at the point $t$. The random variable $I_i^X(t)$ tells us whether the insured is in state $i$ at time $t$ and $N_{ij}^X(t)$ tells us the number of transitions from $i$ to $j$ in the whole period $(0,t)$. 

\begin{defn}[Stochastic cash flow]
A stochastic cash flow is a stochastic process $A=\{A_t\}_{t\geq 0}$ with almost all sample paths with bounded variation.
\end{defn}

More concretely, we will consider cash flows described by an insurance policy entirely determined by its payout functions. We denote by $a_i(t)$, $i\in S$, the sum of payments from the insurer to the insured up to time $t$, given that we know that the insured has always been in state $i$. Moreover, we will denote by $a_{ij}(t)$, $i,j \in S$, $i\neq j$, the payments which are due when the insured switches state from $i$ to $j$ at time $t$. We always assume that these functions are of bounded variation. The cash flows we will consider are entirely described by the policy functions, defined by an insurance policy. Observe that, the policy functions can be stochastic in the case where the payout is linked to a fund modelled by a stochastic process.

\begin{defn}[Policy cash flow]
We consider payout functions $a_i(t)$, $i\in S$ and $a_{ij}(t)$, $i,j\in S$, $i\neq j$ for $t\geq 0$ of bounded variation. The (stochastic) cash flow associated to this insurance is defined by
$$A(t) = \sum_{i\in S} A_i(t) + \sum_{\substack{i,j\in S\\ i\neq j}} A_{ij}(t),$$
where
$$A_i(t) = \int_0^t I_i^X(s) da_i(s), \quad A_{ij}(t) = \int_0^t a_{ij}(s) dN_{ij}^X(s).$$
The quantity $A_i$ corresponds to the accumulated liabilities while the insured is in state $i$ and $A_{ij}$ for the case when the insured switches from $i$ to $j$.
\end{defn}

The value of a stochastic cash flow $A$ at time $t$ will be denoted by $V^+(t,A)$, or simply $V^+(t)$, and is defined as
$$V^+ (t,A)= B_T\int_t^\infty \frac{dA(s)}{B_s},$$
where $B$ is the reference discount factor in \eqref{eq: Money_Market_Account_SDE}. The stochastic integral is a well-defined pathwise Riemann-Stieltjes integral since $A$ is almost surely of bounded variation and $B$ is almost surely continuous. The quantity $V^+(t,A)$ is stochastic since both $B$ and $A$ are stochastic. The prospective reserve of an insurance policy with cash flow $A$ given the information $\mathcal{F}_t$ is then defined as
$$V_\mathcal{F}^+(t,A) = \mathbb{E}^{\mathbb{Q}}[V^+(t,A)|\mathcal{F}_t],$$
where $\mathbb{Q}$ is an equivalent martingale measure.

It turns out, see \cite[Theorem 4.6.3]{Koller}, that one can find explicit expressions when the policy functions $a_i$, $i\in S$, $a_{ij}$, $i,j\in S$, $i\neq j$ and the force of interest are deterministic. Combining the previous result with a conditioning argument allows us to recast the expression for the reserves as the following conditional expectation,
\begin{align}\label{reserve2}
V_i^+(t,A)\triangleq\mathbb{E}^{\mathbb{Q}}\left[\sum_{j\in S} \int_t^\infty \frac{B_t}{B_s}p_{ij}(t,s) da_j(s) + \sum_{\substack{j,k\in S\\k\neq j}} \int_t^\infty \frac{B_t}{B_s} p_{ij}(t,s)\mu_{jk}(s) a_{jk}(s)ds\Bigg| \mathcal{G}_t\right],
\end{align}
where $\mu_{ij}$ are the continuous transition rates associated to the Markov chain $X$ and $p_{ij}(s,t)$ are the transition probabilities from changing from state $i$ at time $s$ to state $j$ at time $t$.

In this paper we will focus on the pricing and hedging of unit-linked pure endowment policies with stochastic volatility and stochastic interest rate.
Since other more general insurances can be reduced to this. For instance, in \eqref{reserve2}, if $a_i$ is of bounded variation and a.e. differentiable with derivative $\dot{a}_i$ then we can look at
$$\mathbb{E}^{\mathbb{Q}}\left[ \frac{B_t}{B_s}\dot{a}_i(s)\Big| \mathcal{G}_t\right] \quad \mbox{and} \quad \mathbb{E}^{\mathbb{Q}}\left[ \frac{B_t}{B_s}a_{ij}(s)\Big| \mathcal{G}_t\right]$$
as contracts with payoff $\dot{a}_i(s)$, respectively $a_{ij}(s)$, with maturity $s\geq t$.


\section{Pricing and hedging of the unit-linked life insurance contract}\label{unit-linked Section}

The aim of this section, is to price and hedge insurance claims linked to the fund $S$. However, we cannot hedge any contingent claim using a portfolio with $S$ only. In the spirit of \cite{RomanoTouzi97} we will complete the market, including the possibility to trade products whose underlying are the forward variance and interest rate, which are indeed actively traded in the market.

\subsection{Completing the market using variance swaps and zero-coupon bonds}\label{subsec: completing}

First, we will introduce a family of equivalent probability measures $\mathbb{Q}\sim\mathbb{P}$ given by
\begin{eqnarray}
\mathbb{Q}\left(A\right) = \mathbb{E}\left[Z_{T}\mathbf{1}_{A}\right],\qquad A\in\mathcal{G}_{T},\label{eq: Q-meas}
\end{eqnarray}
where $Z=\{Z_{t},t\in\left[0,T\right]\}$ is given by
\begin{eqnarray*}
Z_{t} = \mathcal{E}\left(\sum_{i=0}^{2}\int_{0}^{\cdot}\gamma_{s}^{i}dW_{s}^{i}\right)_{t}, \qquad t\in\left[0,T\right],
\end{eqnarray*}
and $\mathcal{G}$-adapted $\gamma^{i},$ for every $i=0,1,2,$ such that $\mathbb{E}\left[Z_{T}\right]=1.$ Here, $\mathcal{E}\left(M\right)_t=\exp\left(M_{t}-\frac{1}{2}\left[M,M\right]_{t}\right)$, denotes the stochastic exponential for a continuous semimartingale $M$.

The following processes are Brownian motions under $\mathbb{Q}$
\begin{align}
W_{t}^{\mathbb{Q},i} & =W_{t}^{i}-\int_{0}^{t}\gamma_{s}^{i}ds, \qquad i=0,1,2. \label{eq: Q-BM_0}
\end{align}
Note that, not all probability measures given in \eqref{eq: Q-meas} are risk-neutral in our market model. In particular, $\gamma^{1}$ is determined by the fact that $S$ is a tradable asset and takes the form
\begin{eqnarray*}
\gamma^{1}_t\triangleq \frac{r_{t}-b\left(t,S_t\right)}{a\left(t,S_t\right)f\left(\nu_{t}\right)}.
\end{eqnarray*}
All probability measures in \eqref{eq: Q-meas} fixing $\gamma^{1}$ are valid risk-neutral measures. In particular, choosing $\gamma^{0}=\gamma^{2}=0$ is one of them. From now on, we denote by $\mathbb{Q}^{0}$ this choice, that is,

\begin{eqnarray}\label{eq: Q_0}
\frac{d\mathbb{Q}^{0}}{d\mathbb{P}} =  \mathcal{E}\left(\int_{0}^{\cdot}\gamma_{s}^{1}dW_{s}^{1}\right)_{T}.
\end{eqnarray}
Now we are in a position to introduce the financial instruments whose valuation will be done under $\mathbb{Q}^{0}$. One of the most traded asset in interest rate markets are zero-coupon bonds.

\begin{defn}
\label{def: Zero-Coupon Bond} A $T$-maturity zero-coupon bond is
a contract that guarantees its holder the payment of one unit of currency
at time $T$, with no intermediate payments. The contract value at
time $0\leq t\leq T$ is denoted by $P_{t,T}$ and by definition $P_{T,T}=1,$
for all $T.$
\end{defn}

A risk-neutral price of a zero-coupon bond in our framework is given in the following definition.
\begin{defn}
\label{def: Zero_Coupon_Bond_Price}The price of a zero-coupon bond,
$P_{t,T}$ is given by
\begin{equation}
P_{t,T}=\mathbb{E}^{\mathbb{Q}^{0}}\left[D_{t,T}\mid\mathcal{G}_{t}\right]=\mathbb{E}^{\mathbb{Q}^{0}}\left[\frac{B_{t}}{B_{T}}\mid\mathcal{G}_{t}\right]=\mathbb{E}^{\mathbb{Q}^{0}}\left[\exp\left\{ -\int_{t}^{T}r_{s}ds\right\} \mid\mathcal{G}_{t}\right],\label{eq: Zero_Coupon_Price}
\end{equation}
where $\mathbb{Q}^{0}$ is the equivalent martingale measure given by \eqref{eq: Q_0}. See \cite[Definition 4.1. in Section 4.3.1 and Section 5.1]{Fili09} for definitions.
\end{defn}

The next classical result gives a connection between the bond price in \eqref{eq: Zero_Coupon_Price} and the solution to a linear PDE, see e.g. 
\cite{Fili09}. 
\begin{lem}\label{ZC-bond_PDE}
Assume
that for any $T>0$, $F_{T}\in\mathcal{C}^{1,2}\left(\left[0,T\right]\times\mathbb{R}\right)$
is a solution to the boundary problem on $\left[0,T\right]\times\mathbb{R}$
given by
\begin{align*}
\partial_{t}F_{T}\left(t,x\right)+\mu\left(t,x\right)\partial_{x}F_{T}\left(t,x\right)
+\frac{1}{2}\sigma^{2}\left(t,x\right)\partial_{x}^{2}F_{T}\left(t,x\right)-xF_{T}\left(t,x\right) & =0,\\
F_{T}\left(T,x\right) & =1.
\end{align*}
Then 
\[
M_t\triangleq F_{T}\left(t,r_{t}\right)e^{-\int_{0}^{t}r_{u}du},\qquad t\in\left[0,T\right],
\]
is a local martingale. If in addition either:
\end{lem}

\begin{itemize}
\item[(a)] $\mathbb{E}^{\mathbb{Q}^{0}}\left[\int_{0}^{T}\left|\partial_{x}F_{T}\left(t,r_{t}\right)e^{-\int_{0}^{t}r_{u}du}\sigma\left(t,r_{t}\right)\right|^{2}dt\right]<\infty,$
or
\item[(b)] $M$ is uniformly bounded,
\end{itemize}
then $M$ is a martingale, and 
\begin{equation}
F_{T}\left(t,r_{t}\right)=\mathbb{E}^{\mathbb{Q}^{0}}\left[e^{-\int_{t}^{T}r_{u}du}\mid\mathcal{G}_{t}\right],\qquad t\in\left[0,T\right].\label{eq: pseudo- ZCbond_Price}
\end{equation}

The dynamics of the zero-coupon bond $P$ in terms of the function $F_{T}$ are given by 
\begin{align}
dP_{t,T} & =\mathcal{L}_{P}\left(F_{T}\right)\left(t,r_{t}\right)dt+\partial_{x}F_{T}\left(t,r_{t}\right)\sigma\left(t,r_{t}\right)dW_{t}^{0},\label{eq: ZC_P-dynamics}
\end{align}
where $\mathcal{L}_{P}\triangleq\partial_{t}+\mu\left(t,x\right)\partial_{x}+\frac{1}{2}\sigma^{2}\left(t,x\right)\partial_{x}^{2}-x$.


We turn now to the definition of the forward variance process. The forward variance $\xi_{t,u}$, for $0\leq t\leq u$, is by definition the conditional expectation of the future instantaneous variance, see e.g. \cite{Aly14}, that is,
\begin{eqnarray}
\xi_{t,u}\triangleq\mathbb{E}^{\mathbb{Q}^{0}}\left[\nu_{u}\mid\mathcal{G}_{t}\right],\quad 0\leq t\leq u,\label{FwdVar_vs_InstaVar}
\end{eqnarray}
where $\mathbb{Q}^{0}$ is the risk-neutral pricing measure defined in \eqref{eq: Q_0}.
Following \cite{BergomiGuyon2012}, one can easily rewrite the general stochastic volatility model, given by equations $\left(\ref{eq: Stock_Price_SDE}\right)$ and $\left(\ref{eq: Instantaneous_Variance}\right)$ in forward variance form. This is achieved by taking conditional expectation of equation $\left(\ref{eq: Instantaneous_Variance}\right)$, which yields

\begin{eqnarray*}
d\mathbb{E}^{\mathbb{Q}^{0}}\left[\nu_u\mid\mathcal{G}_t\right] &=& -\kappa\left(\mathbb{E}^{\mathbb{Q}^{0}}\left[
\nu_u\mid\mathcal{G}_t\right]-\bar{\nu}\right)du,\quad u>t,
\end{eqnarray*}
Solving the previous linear ODE, by integrating on $\left[t,u\right]$, we have
\begin{eqnarray}
\xi_{t,u} &=& \bar{\nu} + e^{-\kappa\left(u-t\right)}\left(\nu_t-\bar{\nu}\right)\label{FwdVar_InstaVar_RawRelationship}.
\end{eqnarray}
There are two things to notice at this point. The first is that, by construction $\nu_t=\xi_{t,t}$, for every $t\in\left[0,T\right]$. 
Second is that differentiating the previous equation, we can characterize the dynamics with respect to $t$ for the forward variance as follows
\begin{eqnarray}
d\xi_{t,u} &=& e^{-\kappa\left(u-t\right)}h\left(\nu_t\right)dW_t^2.\label{Fwd_Var_dynamics}
\end{eqnarray}
Solving equation $\left(\ref{FwdVar_InstaVar_RawRelationship}\right)$ for $\nu_t$, yields
\begin{eqnarray*}
\nu_t &=& \bar{\nu} + e^{\kappa \left(u-t\right)}\left(\xi_{t,u}-\bar{\nu} \right)\triangleq\psi\left(t,u,\xi_{t,u}\right).\\
\end{eqnarray*}
Usually, the dynamics of the forward variance in any forward variance model, are given through the following SDE,
\begin{eqnarray}
d\xi_{t,u} & =\lambda\left(t,u,\xi_{t,u}\right)dW_{t}^{2}.\label{eq: Fwd_Variance_SDE}
\end{eqnarray}
As a consequence of the previous result and in our case, the function $\lambda$ in equation $\left(\ref{eq: Fwd_Variance_SDE}\right)$ is fully characterized by
\begin{eqnarray}
\lambda\left(t,u,\xi_{t,u}\right) &\triangleq& e^{-\kappa\left(u-t\right)}\left(h\circ \psi\right)\left(t,u,\xi_{t,u}\right). \label{eq: Lambda}
\end{eqnarray}

Note that any finite-dimensional Markovian
stochastic volatility model can be rewritten in forward variance form. Since we will only be interested in the fixed case $u=T$, we will
drop the dependence on $T$ for $\xi_{t,T}$ and write instead $\xi_{t}=\xi_{t,T}$.

We will show how to form a portfolio with a perfect hedge.
The financial instruments needed in order to build a riskless portfolio
are the underlying asset, a variance swap and the zero-coupon bond.

From now on, we will assume that the function $F_T$, solution to the PDE in Lemma \ref{ZC-bond_PDE} is invertible in the space variable, for every $t\in\left[0,T\right],$ e.g. this is the case if $r_t$, $t\in\left[0,T\right]$ is given by the Vasicek model. Introduce the notation,
\begin{eqnarray}\label{eq: G_partialF}
G_{T}\left(t,x\right)\triangleq\partial_{x}F_{T}\left(t,x\right),
\end{eqnarray} 
then $\partial_{x}F_{T}\left(t,r_{t}\right)=G_{T}\left(t,F_{T}^{-1}\left(t,P_{t,T}\right)\right)$, where $r_{t} = F_{T}^{-1}\left(t,P_{t,T}\right).$

\subsection{Pricing and hedging in the completed market}

Let $\Pi=\left\{ \Pi_{t}\right\} _{t\in\left[0,T\right]}$ be a stochastic
process representing the value of a portfolio consisting of a long
position on an option with price $V_{t}$, where $V_{t}=V\left(t,S_{t},\xi_{t},P_{t,T}\right)$, and respective short positions on $\Delta_{t}$ units of the underlying asset, $\Sigma_{t}$ units of a variance swap, and $\Psi_{t}$ units
of a zero-coupon bond. Therefore, we can characterize the process $\Pi$
as 
\begin{equation}
\Pi_{t}=V\left(t,S_{t},\xi_{t},P_{t,T}\right)-\Delta_{t}S_{t}-\Sigma_{t}\xi_{t}-\Psi_{t}P_{t,T},\quad t\in\left[0,T\right].\label{eq: Portfolio}
\end{equation}

\begin{defn}
We say that the portfolio $\Pi$ is self-financing if, and only if,
\[
d\Pi_{t}=dV\left(t,S_{t},\xi_{t},P_{t,T}\right)-\Delta_{t}dS_{t}-\Sigma_{t}d\xi_{t}-\Psi_{t}dP_{t,T},
\]
for every $t\in\left[0,T\right].$
\end{defn}

\begin{defn}
\label{def: Self-financing Portfolio}We say that the portfolio $\Pi$ is perfectly hedged, or risk-neutral, if it is self-financing and
\[
\Pi_{T}=0.
\]
\end{defn}

From now on, and throughout the rest of the paper, we will only differentiate
between time derivative $\partial_{t}V$ and space derivatives
$\partial_{x}V$, $\partial_{y}V$, $\partial_{z}V$,
to write the partial derivatives of $V=V\left(t,x,y,z\right)$. We
will also denote second order spatial partial derivatives of $V$
with respect to $S_{t}$, $\xi_{t}$, $P_{t,T}$, respectively by
$\partial_{x}^{2}V$, $\partial_{y}^{2}V$, $\partial_{z}^{2}V$
and the second order crossed derivatives as $\partial_{x}\partial_{y}V$,
$\partial_{x}\partial_{z}V$, $\partial_{y}\partial_{z}V$. 
In order to simplify the notation in the following results, we shall define
\[\Xi_{T}\left(t,x\right)\triangleq G_{T}\left(t,F_{T}^{-1}\left(t,x\right)\right)\cdot\sigma\left(t,F_{T}^{-1}\left(t,x\right)\right),\]
where recall that $G_{T}$ is given in \eqref{eq: G_partialF}. 
\begin{thm}
Let $\Pi$ be a portfolio defined as in $\left(\ref{eq: Portfolio}\right)$,
and assume $V\in C^{1,2}\left(\left[0,T\right]\times\mathbb{R}^{3}\right)$.
If $\Pi$ is a replicating portfolio, then $V$ fulfils,
\begin{align}
\partial_{t}V+\frac{1}{2}\left(x^{2}a\left(t,x\right)^{2}f\left(\psi\left(t,T,y\right)\right)^{2}\partial_{x}^{2}V+\lambda\left(t,T,y\right)^{2}\partial_{y}^{2}V+\Xi_{T}^{2}\left(t,z\right)\partial_{z}^{2}V\right)\label{eq: Pricing_PDE}\\ 
-r_{t}\left(V-x\partial_{x}V-y\partial_{y}V-z\partial_{z}V\right) & =0,\nonumber 
\end{align}
for every $t\in\left[0,T\right]$ and
\begin{equation}
V\left(T,S_{T},\xi_{T},P_{T,T}\right)=\max\left(S_{T},G\right).\label{eq: PDE_Terminal_Condition}
\end{equation}
\end{thm}

\begin{proof}
It is important to notice that we will use the notation $V_{t}$ to refer to the process $V\left(t,S_t,\xi_t,P_{t,T}\right)$, and similarly for the partial derivatives. For instance, $\partial_x V_t=\partial_x V\left(t,S_t,\xi_t,P_{t,T}\right)$. By means of Itô's lemma, we are able to write the change in our portfolio $\left\{V_t\right\}_{t\in\left[0,T\right]}$ as follows
\begin{align*}
d\Pi_{t} & =\partial_{t}V_{t}dt+\partial_{x}V_{t}dS_{t}+\partial_{y}V_{t}d\xi_{t}+\partial_{z}V_{t}dP_{t,T}\\
 & \qquad+\frac{1}{2}\partial_{x}^{2}V_{t}d\left[ S,S\right]_{t}+\frac{1}{2}\partial_{y}^{2}V_{t}d\left[\xi,\xi\right]_{t}+\frac{1}{2}\partial_{z}^{2}V_{t}d\left[ P,P\right]_{t}\\
 & \qquad+\partial_{x}\partial_{y}V_{t}d\left[ S,\xi\right]_{t}+\partial_{x}\partial_{z}V_{t}d\left[ S,P\right]_{t}+\partial_{y}\partial_{z}V_{t}d\left[\xi,P\right]_{t}\\
 & \qquad-\Delta_{t}dS\left(t\right)-\Sigma_{t}d\xi_{t}-\Psi_{t}dP_{t,T}\\
 & =\partial_{t}V_{t}dt+\left\{ \partial_{x}V_{t}-\Delta_{t}\right\} dS_{t}+\left\{ \partial_{y}V_{t}-\Sigma_{t}\right\} d\xi_{t}+\left\{ \partial_{z}V_{t}-\Psi_{t}\right\} dP_{t,T}\\
 & \qquad+\frac{1}{2}\partial_{x}^{2}V_{t}d\left[ S,S\right]_{t}+\frac{1}{2}\partial_{y}^{2}V_{t}d\left[\xi,\xi\right]_{t}+\frac{1}{2}\partial_{z}^{2}V_{t}d\left[ P,P\right]_{t}\\
 & \qquad+\partial_{x}\partial_{y}V_{t}d\left[ S,\xi\right]_{t}+\partial_{x}\partial_{z}V_{t}d\left[ S,P\right]_{t}+\partial_{y}\partial_{z}V_{t}d\left[\xi,P\right]_{t}.
\end{align*}
Using the dynamics for $dS_{t}$, $d\xi_{t}$, $dP_{t,T}$ and the quadratic covariations, given by
\begin{align*}
d\left[ S,S\right]_{t} & =S_{t}^{2}a\left(t,S_{t}\right)^{2}f\left(\psi\left(t,T,\xi_t\right)\right)^{2}dt,\\
d\left[\xi,\xi\right]_{t} & =\lambda\left(t,T,\xi_{t}\right)^{2}dt,\\
d\left[ P,P\right]_{t} & =\Xi_{T}^{2}\left(t,P_{t,T}\right)dt,\\
d\left[ S,\xi\right]_{t} & = 0,\\
d\left[ S,P\right]_{t} & = 0,\\
d\left[\xi,P\right]_{t} & = 0,
\end{align*}
we obtain
\begin{align*}
d\Pi_{t} & =\partial_{t}V_{t}dt+\left\{ \partial_{x}V_{t}-\Delta_{t}\right\} dS_{t}+\left\{ \partial_{y}V_{t}-\Sigma_{t}\right\} d\xi_{t}+\left\{ \partial_{z}V_{t}-\Psi_{t}\right\} dP_{t,T}\\
 & \qquad+\frac{1}{2}S_{t}^{2}a\left(t,S_{t}\right)^{2}f\left(\psi\left(t,T,\xi_t\right)\right)^{2}\partial_{x}^{2}V_{t}dt+\frac{1}{2}\lambda\left(t,T,\xi_{t}\right)^{2}\partial_{y}^{2}V_{t}dt+\frac{1}{2}\Xi_{T}^{2}\left(t,P_{t,T}\right)\partial_{z}^{2}V_{t}dt\\
 & =\Biggl\{\partial_{t}V_{t}+\frac{1}{2}\left(S_{t}^{2}a\left(t,S_{t}\right)^{2}f\left(\psi\left(t,T,\xi_t\right)\right)^{2}\partial_{x}^{2}V_{t}+\lambda\left(t,T,\xi_{t}\right)^{2}\partial_{y}^{2}V_{t}+\Xi_{T}^{2}\left(t,P_{t,T}\right)\partial_{z}^{2}V_{t}\right)\Biggr\} dt\\
 & \qquad+\left\{ \partial_{x}V_{t}-\Delta_{t}\right\} dS_{t}+\left\{ \partial_{y}V_{t}-\Sigma_{t}\right\} d\xi_{t}+\left\{ \partial_{z}V_{t}-\Psi_{t}\right\} dP_{t,T}.
\end{align*}

Now, in order to make the portfolio instantaneously risk-free, we must
impose that the return on our portfolio equals the risk-free rate
$r_{t}$, i.e. $d\Pi_{t}=r_{t}\Pi_{t}dt=r_{t}\left(V_{t}-\Delta_{t}S_{t}-\Sigma_{t}\xi_{t}-\Psi_{t}P_{t,T}\right)dt$,
and force the coefficients in front of $dS_{t}$, $d\xi_{t}$ and
$dP_{t,T}$ to be zero, i.e.
\begin{align*}
\Delta_{t} & =\partial_{x}V_{t},\\
\Sigma_{t} & =\partial_{y}V_{t},\\
\Psi_{t} & =\partial_{z}V_{t}.
\end{align*}
This implies that
\begin{align*}
d\Pi_{t} & =\Biggl\{ \partial_{t}V_{t}+\frac{1}{2}\left(S_{t}^{2}a\left(t,S_{t}\right)^{2}f\left(\psi\left(t,T,\xi_t\right)\right)^{2}\partial_{x}^{2}V_{t}+\lambda\left(t,T,\xi_{t}\right)^{2}\partial_{y}^{2}V_{t}+\Xi_{T}^{2}\left(t,P_{t,T}\right)\partial_{z}^{2}V_{t}\right)\Biggr\} dt.\\
\end{align*}
Therefore, rearranging the terms in the previous expression and taking
into account that we have imposed $\Delta_{t}=\partial_{x}V_{t}$,
$\Sigma_{t}=\partial_{y}V_{t}$, $\Psi_{t}=\partial_{z}V_{t}$, we
have the PDE for the unit-linked product, ending the proof.
\end{proof}

From now on, in order to ease the notation, we will define the
differential operator in $\left(\ref{eq: Pricing_PDE}\right)$ as
\begin{align}\label{eq: L_V}
\mathcal{L}_{V} & \triangleq\partial_{t}+\frac{1}{2}\left(x^{2}a\left(t,x\right)^{2}f\left(\psi\left(t,T,y\right)\right)^{2}\partial_{x}^{2}+\lambda\left(t,T,y\right)^{2}\partial_{y}^{2}+\Xi_{T}^{2}\left(t,z\right)\partial_{z}^{2}\right)\\
 &\qquad -r_{t}\left(1-x\partial_{x}-y\partial_{y}-z\partial_{z}\right). \nonumber
\end{align}
We will now prove that the discounted option price is a martingale.
\begin{thm}
\label{thm: Generalized_UL-Pricing_Thm}Let $V$
be the solution to the PDE given by equation $\left(\ref{eq: Pricing_PDE}\right)$
with terminal condition $\left(\ref{eq: PDE_Terminal_Condition}\right)$.
Then
\[
B_{t}^{-1}V\left(t,S_{t},\xi_{t},P_{t,T}\right)=\mathbb{E}^{\mathbb{Q}}\left[B_{T}^{-1}V\left(T,S_{T},\xi_{T},P_{T,T}\right)\mid\mathcal{G}_{t}\right],
\]
where $\mathbb{Q}$ indicates the risk-neutral measure. 
\end{thm}

\begin{proof}
We start by imposing that the discounted price process, $\tilde{S}_{t}=B_{t}^{-1}S_{t},$
the discounted variance swap $\tilde{\xi_{t}},$ and the discounted
zero-coupon bond price $\tilde{P}_{t,T}$ are $\mathbb{Q}-$martingales,
where $dB_{t}=r_{t}B_{t}dt$ and $dB_{t}^{-1}=-r_{t}B_{t}^{-1}dt$. To do so, we will also make use of the relationship between the Brownian motions and their $\mathbb{Q}$-measure counterparts, given by $\left(\ref{eq: Q-BM_0}\right)$. 
\begin{align*}
d\tilde{S}_{t} & =dB_{t}^{-1}S_{t}+B_{t}^{-1}dS_{t}\\
 & =-r_{t}B_{t}^{-1}S_{t}dt+B_{t}^{-1}\left[b\left(t,S_{t}\right)S_{t}dt+a\left(t,S_{t}\right)f\left(\psi\left(t,T,\xi_t\right)\right)S_{t}dW_{t}^{1}\right]\\
 & =\tilde{S_{t}}\left[\left(b\left(t,S_{t}\right)-r_{t}\right)dt+a\left(t,S_{t}\right)f\left(\psi\left(t,T,\xi_t\right)\right)\left[dW_{t}^{\mathbb{Q},1}
+\gamma_{t}^{1}dt\right]\right]\\
 & =\tilde{S_{t}}\left[b\left(t,S_{t}\right)-r_{t}+a\left(t,S_{t}\right)f\left(\psi\left(t,T,\xi_t\right)\right)\gamma_{t}^{1}\right]dt
+\tilde{S_{t}}a\left(t,S_{t}\right)f\left(\nu_{t}\right)dW_{t}^{\mathbb{Q},1}.
\end{align*}
Now, the discounted price process $\tilde{S_{t}}$ is a $\mathbb{Q}-$martingale
if, and only if
\begin{equation}
\gamma_{t}^{1} = \frac{r_{t}-b\left(t,S_{t}\right)}{a\left(t,S_{t}\right)f\left(\psi\left(t,T,\xi_t\right)\right)}.\label{eq: Price_Q-Martingale_Condition}
\end{equation}
We do the same for the discounted forward variance process,
\begin{align*}
d\tilde{\xi}_{t} & =dB_{t}^{-1}\xi_{t}+B_{t}^{-1}d\xi_{t}\\
 & =-r_{t}B_{t}^{-1}\xi_{t}dt+B_{t}^{-1}\lambda\left(t,T,\xi_{t}\right)dW_{t}^{2}\\
 & =B_{t}^{-1}\left[\lambda\left(t,T,\xi_{t}\right)\gamma_{t}^{2}-r_{t}\xi_{t}\right]dt
 +B_{t}^{-1}\lambda\left(t,T,\xi_{t}\right)dW_{t}^{\mathbb{Q},2}.
\end{align*}
Therefore, the discounted variance swap is a $\mathbb{Q}-$martingale
if, and only if
\begin{equation}
\gamma_{t}^{2}=\frac{r_{t}\xi_{t}}{\lambda\left(t,T,\xi_{t}\right)}.\label{eq: FwdVarSwap_Q-Martingale Condition}
\end{equation}
Finally, we impose that the discounted zero-coupon bond price process
is a $\mathbb{Q}$-martingale analogously
\begin{align*}
d\tilde{P}_{t,T} & =dB_{t}^{-1}P_{t,T}+B_{t}^{-1}dP_{t,T}\\
 & =-r_{t}B_{t}^{-1}P_{t,T}dt\\
 & \qquad+B_{t}^{-1}\biggl[\partial_{t}F_{T}\left(t,r_{t}\right)+\mu\left(t,r_{t}\right)\partial_{x}F_{T}\left(t,r_{t}\right) \\
 & \qquad +\frac{1}{2}\sigma^{2}\left(t,r_{t}\right)\partial_{x}^{2}F_{T}\left(t,r_{t}\right)-r_{t}F_{T}\left(t,r_{t}\right)\biggr]dt
 +B_{t}^{-1}\partial_{x}F_{T}\left(t,r_{t}\right)\sigma\left(t,r_{t}\right)dW_{t}^{0}\\
 & =B_{t}^{-1}\biggl[\partial_{t}F_{T}\left(t,r_{t}\right)+\left(\mu\left(t,r_{t}\right)+\sigma\left(t,r_{t}\right)\gamma_{t}^{0}\right)\partial_{x}F_{T}\left(t,r_{t}\right) \\
 & \qquad +\frac{1}{2}\sigma^{2}\left(t,r_{t}\right)\partial_{x}^{2}F_{T}\left(t,r_{t}\right)-r_{t}\left(F_{T}\left(t,r_{t}\right)+P_{t,T}\right)\biggr]dt\\
 & \qquad+B_{t}^{-1}\Xi_{T}^{2}\left(t,P_{t,T}\right)dW_{t}^{\mathbb{Q},0}.
\end{align*}
Therefore, the discounted zero-coupon bond is a $\mathbb{Q}$-martingale
if, and only if
\begin{eqnarray}\label{eq: ZCbond_Q-Martingale_Condition}
\frac{1}{F_{T}\left(t,r_{t}\right)+P_{t,T}}\biggl[\partial_{t}F_{T}\left(t,r_{t}\right)+\left(\mu\left(t,r_{t}\right)+\sigma\left(t,r_{t}\right)\gamma_{t}^{0}\right)\partial_{x}F_{T}\left(t,r_{t}\right)
+\frac{1}{2}\Xi_{T}^{2}\left(t,P_{t,T}\right)\biggr]=r_t.
\end{eqnarray}
Now, we are able to characterize $\gamma^{i}$, for all $i\in\left\{ 0,1,2\right\} ,$
by solving the linear system given by equations $\left(\ref{eq: Price_Q-Martingale_Condition}\right)$,
$\left(\ref{eq: FwdVarSwap_Q-Martingale Condition}\right)$ and $\left(\ref{eq: ZCbond_Q-Martingale_Condition}\right)$.

Therefore, we will apply Itô's lemma to the discounted price of the
option,
\[
d\left[B_{t}^{-1}V\left(t,S_{t},\xi_{t},P_{t,T}\right)\right]=dB_{t}^{-1}V\left(t,S_{t},\xi_{t},P_{t,T}\right)+B_{t}^{-1}dV\left(t,S_{t},\xi_{t},P_{t,T}\right).
\]
In order to relax the notation we will drop the dependencies of $V$,
allowing us to rewrite the previous expression as

\begin{align*}
d\left[B_{t}^{-1}V_{t}\right] & =dB_{t}^{-1}V_{t}+B_{t}^{-1}dV_{t}\\
 & =-r_{t}B_{t}^{-1}V_{t}dt
 +B_{t}^{-1}\left[\partial_{t}V_{t}dt+\partial_{x}V_{t}dS_{t}+\partial_{y}V_{t}d\xi_{t}+\partial_{z}V_{t}dP_{t,T}\right]\\
 & \qquad+B_{t}^{-1}\left[\frac{1}{2}\partial_{x}^{2}V_{t}d\left[ S,S\right]_{t}+\frac{1}{2}\partial_{y}^{2}V_{t}d\left[\xi,\xi\right]_{t}+\frac{1}{2}\partial_{z}^{2}V_{t}d\left[ P,P\right]_{t}\right]\\
 & \qquad+B_{t}^{-1}\left[\partial_{x}\partial_{y}V_{t}d\left[ S,\xi\right]_{t}+\partial_{x}\partial_{z}V_{t}d\left[ S,P\right]_{t}+\partial_{y}\partial_{z}V_{t}d\left[\xi,P\right]_{t}\right].\\
\end{align*}
Furthermore,
\begin{align*}
d\left[B_{t}^{-1}V_{t}\right] & =B_{t}^{-1}\left(\partial_{t}V_{t}-r_{t}V_{t}\right)dt\\
 & \quad+B_{t}^{-1}\left(\partial_{x}V_{t}\left[b\left(t,S_{t}\right)S_{t}dt+S_{t}a\left(t,S_{t}\right)f\left(\psi\left(t,T,\xi_t\right)\right)dW_{t}^{1}\right]+\partial_{y}V_{t}\left[\lambda\left(t,T,\xi_{t}\right)dW_{t}^{2}\right]\right)\\
 & \quad+B_{t}^{-1}\partial_{z}V_{t}\left[\mathcal{L}_{P}\left(F_{T}\left(t,r_{t}\right)\right)dt+\Xi_{T}\left(t,P_{t,T}\right)dW_{t}^{0}\right]\\
 & \quad+B_{t}^{-1}\left[\frac{1}{2}S_{t}^{2}a\left(t,S_{t}\right)^{2}f\left(\psi\left(t,T,\xi_t\right)\right)^{2}\partial_{x}^{2}V_{t}+\frac{1}{2}\lambda\left(t,T,\xi_{t}\right)^{2}\partial_{y}^{2}V_{t}+\frac{1}{2}\Xi_T^{2}\left(t,P_{t,T}\right)\partial_{z}^{2}V_{t}\right]dt\\
 & =B_{t}^{-1}\biggl[\partial_{t}V_{t}-r_{t}V_{t}+b\left(t,S_{t}\right)S_{t}\partial_{x}V_{t}+\mathcal{L}_{P}\left(F_{T}\left(t,r_{t}\right)\right)\partial_{z}V_{t} \\
 & \quad\qquad+\frac{1}{2}\left(S_{t}^{2}a\left(t,S_{t}\right)^{2}f\left(\psi\left(t,T,\xi_t\right)\right)^{2}\partial_{x}^{2}V_{t}+\lambda\left(t,T,\xi_{t}\right)^{2}\partial_{y}^{2}V_{t}+\Xi_{T}^{2}\left(t,P_{t,T}\right)\partial_{z}^{2}V_{t}\right)\biggr]dt\\
 & \quad+B_{t}^{-1}\left[S_{t}f\left(\psi\left(t,T,\xi_t\right)\right)\partial_{x}V_{t}dW_{t}^{1}+\lambda\left(t,T,\xi_{t}\right)\partial_{y}V_{t}dW_{t}^{2}+\Xi_{T}\left(t,P_{t,T}\right)\partial_{z}V_{t}dW_{t}^{0}\right].
\end{align*}

If we replace the Brownian motions under the $\mathbb{P}$-measure
by the ones under the $\mathbb{Q}$-measure given by equation $\left(\ref{eq: Q-BM_0}\right)$,
we can rewrite the previous expression as follows
\begin{align*}
& d\left[B_{t}^{-1}V_{t}\right]= \\
& \quad =B_{t}^{-1}\biggl[\partial_{t}V_{t}-r_{t}V_{t}+b\left(t,S_{t}\right)S_{t}\partial_{x}V_{t}+\mathcal{L}_{P}\left(F_{T}\left(t,r_{t}\right)\right)\partial_{z}V_{t} \\
 & \qquad\qquad+\frac{1}{2}\left(S_{t}^{2}a\left(t,S_{t}\right)^{2}f\left(\psi\left(t,T,\xi_t\right)\right)^{2}\partial_{x}^{2}V_{t}+\lambda\left(t,T,\xi_{t}\right)^{2}\partial_{y}^{2}V_{t}+\Xi_{T}^{2}\left(t,P_{t,T}\right)\partial_{z}^{2}V_{t}\right)\biggr]dt\\
 & \qquad+B_{t}^{-1}S_{t}a\left(t,S_{t}\right)f\left(\psi\left(t,T,\xi_t\right)\right)\partial_{x}V_{t}\left[dW_{t}^{\mathbb{Q},1}+\gamma_{t}^{1}dt\right]\\
 & \qquad+B_{t}^{-1}\lambda\left(t,T,\xi_{t}\right)\partial_{y}V_{t}\left[dW_{t}^{\mathbb{Q},2}+\gamma_{t}^{2}dt\right]\\
 & \qquad+B_{t}^{-1}\Xi_{T}\left(t,P_{t,T}\right)\partial_{z}V_{t}\left[dW_{t}^{\mathbb{Q},0}+\gamma_{t}^{0}dt\right]\\
 & =B_{t}^{-1}\biggl[\partial_{t}V_{t}-r_{t}V_{t}+S_{t}\left[b\left(t,S_{t}\right)+\gamma_{t}^{1}a\left(t,S_{t}\right)f\left(\psi\left(t,T,\xi_t\right)\right)\right]\partial_{x}V_{t} \\
 & \qquad\qquad+\gamma_{t}^{2}\lambda\left(t,T,\xi_{t}\right)\partial_{y}V_{t}
 +\left[\mathcal{L}_{P}\left(F_{T}\left(t,r_{t}\right)\right)+\gamma_{t}^{0}\Xi_{T}\left(t,P_{t,T}\right)\right]\partial_{z}V_{t}\\
 & \qquad\qquad+\frac{1}{2}\left(S_{t}^{2}a\left(t,S_{t}\right)^{2}f\left(\psi\left(t,T,\xi_t\right)\right)^{2}\partial_{x}^{2}V_{t}+\lambda\left(t,T,\xi_{t}\right)^{2}\partial_{y}^{2}V_{t}+\Xi_{T}^{2}\left(t,P_{t,T}\right)\partial_{z}^{2}V_{t}\right)\biggr]dt\\
 & \qquad+B_{t}^{-1}\left[S_{t}a\left(t,S_{t}\right)f\left(\psi\left(t,T,\xi_t\right)\right)\partial_{x}V_{t}dW_{t}^{\mathbb{Q},1}+\lambda\left(t,T,\xi_{t}\right)\partial_{y}V_{t}dW_{t}^{\mathbb{Q},2}+\Xi_{T}\left(t,P_{t,T}\right))\partial_{z}V_{t}dW_{t}^{\mathbb{Q},0}\right].
\end{align*}
Applying equations $\left(\ref{eq: Price_Q-Martingale_Condition}\right)$,
$\left(\ref{eq: FwdVarSwap_Q-Martingale Condition}\right)$, $\left(\ref{eq: ZCbond_Q-Martingale_Condition}\right)$
and reorganizing the terms in the previous equation, we have
\begin{align*}
 & d\left[B_{t}^{-1}V_{t}\right]=
 B_{t}^{-1}\biggl[\partial_{t}V_{t}+r_{t}\left(S_{t}\partial_{x}V_{t}+\xi_{t}\partial_{y}V_{t}+P_{t,T}\partial_{z}V_{t}-V_{t}\right)\\
 & \qquad\qquad+\frac{1}{2}\left(S_{t}^{2}a\left(t,S_{t}\right)^{2}f\left(\psi\left(t,T,\xi_t\right)\right)^{2}\partial_{x}^{2}V_{t}+\lambda\left(t,T,\xi_{t}\right)^{2}\partial_{y}^{2}V_{t}+\Xi_{T}^{2}\left(t,P_{t,T}\right)\partial_{z}^{2}V_{t}\right)\biggr]dt\\
 & \qquad+B_{t}^{-1}\left[S_{t}a\left(t,S_{t}\right)f\left(\psi\left(t,T,\xi_t\right)\right)\partial_{x}V_{t}dW_{t}^{\mathbb{Q},1}+\lambda\left(t,T,\xi_{t}\right)\partial_{y}V_{t}dW_{t}^{\mathbb{Q},2}+\Xi_{T}\left(t,P_{t,T}\right)\partial_{z}V_{t}dW_{t}^{\mathbb{Q},0}\right].
\end{align*}
Now, noticing that the $dt$ term in the previous equation is the differential operator $\eqref{eq: L_V}$ applied to $V$, we can write the following
\begin{align*}
d\left[B_{t}^{-1}V\left(t,S_{t},\xi_{t},P_{t,T}\right)\right] & =B_{t}^{-1}\mathcal{L}_{V}V\left(t,S_{t},\xi_{t},P_{t,T}\right)dt\\
 & \qquad+B_{t}^{-1}S_{t}a\left(t,S_{t}\right)f\left(\psi\left(t,T,\xi_t\right)\right)\partial_{x}V\left(t,S_{t},\xi_{t},P_{t,T}\right)dW_{t}^{\mathbb{Q},1}\\
 & \qquad+B_{t}^{-1}\lambda\left(t,u,\xi_{t}\right)\partial_{y}V\left(t,S_{t},\xi_{t},P_{t,T}\right)dW_{t}^{\mathbb{Q},2}\\
 & \qquad+B_{t}^{-1}\Xi_{T}\left(t,P_{t,T}\right)\partial_{z}V\left(t,S_{t},\xi_{t},P_{t,T}\right)dW_{t}^{\mathbb{Q},0}.
\end{align*}
Next, integrating on the interval $\left[s,t\right],$ with $s\leq t,$
we can write the previous equation in integral form as
\begin{align*}
B_{t}^{-1}V\left(t,S_{t},\xi_{t},P_{t,T}\right) & =V\left(s,S_{s},\xi_{s},P_{s,T}\right)+\int_{s}^{t}B_{\tau}^{-1}\mathcal{L}_{V}V\left(\tau,S_{\tau},\xi_{\tau},P_{\tau,T}\right)d\tau\\
 & \qquad+\int_{s}^{t}B_{\tau}^{-1}S_{\tau}a\left(\tau,S_{\tau}\right)f\left(\psi\left(\tau,T,\xi_\tau\right)\right)\partial_{x}V\left(\tau,S_{\tau},\xi_{\tau},P_{\tau,T}\right)dW_{\tau}^{\mathbb{Q},1}\\
 & \qquad+\int_{s}^{t}B_{\tau}^{-1}\lambda\left(\tau,u,\xi_{\tau}\right)\partial_{y}V\left(\tau,S_{\tau},\xi_{\tau},P_{\tau,T}\right)dW_{\tau}^{\mathbb{Q},2}\\
 & \qquad+\int_{s}^{t}B_{\tau}^{-1}\Xi_{T}\left(t,P_{t,T}\right)\partial_{z}V\left(\tau,S_{\tau},\xi_{\tau},P_{\tau,T}\right)dW_{\tau}^{\mathbb{Q},0}.
\end{align*}
Taking the conditional expectation with respect to the risk neutral
measure, we have that
\[
\mathbb{E}^{\mathbb{Q}}\left[B_{t}^{-1}V\left(t,S_{t},\xi_{t},P_{t,T}\right)\mid\mathcal{G}_{s}\right]=V_{s}+\mathbb{E}^{\mathbb{Q}}\left[\int_{s}^{t}B_{\tau}^{-1}\mathcal{L}_{V}V\left(\tau,S_{\tau},\xi_{\tau},P_{\tau,T}\right)d\tau\mid\mathcal{G}_{s}\right].
\]
Notice that the previous expression is a martingale if, and only if,
$\mathcal{L}_{V}V\left(t,S_{t},\xi_{t},P_{t,T}\right)\equiv0,$ for
all $t\in\left[0,T\right].$ 
\end{proof}


\section{The Vasicek model and Heston model written in
forward variance}\label{sec: 4}

This section is devoted to providing the reader with a particular
model. We will assume that the evolution
of the short-term rate is given by a Vasicek model and consider a
Heston model for the risky asset written in forward variance
form.

Let us consider the following SDE for the short-term rate given by
the Vasicek model.

\begin{equation}
dr_{t}=k\left(\theta-r_{t}\right)dt+\sigma dW_{t}^{0},\quad r_0>0,\quad t\in\left[0,T\right],\label{eq: Vasicek_Model}
\end{equation}
and the Heston model for the risky asset, given by
\begin{align}
dS_{t} & =\mu_{t}S_{t}dt+S_{t}\sqrt{\nu_{t}}dW_{t}^{1},\quad S_0>0,\quad t\in\left[0,T\right],\label{eq: Heston_Price}\\
d\nu_{t} & =-\kappa\left(\nu_{t}-\bar{\nu}\right)dt+\eta\sqrt{\nu_{t}}dW_{t}^{2},\quad \nu_0>0, \quad t\in\left[0,T\right].\label{eq: Heston_instantaneous_variance}
\end{align}
It is well known that the SDE $\left(\ref{eq: Vasicek_Model}\right)$ admits the following closed expression. 
\begin{align*}
r_{T} & =e^{-k\left(T-t\right)}r_{t}+\theta\left(1-e^{-k\left(T-t\right)}\right)+\sigma\int_{t}^{T}e^{-k\left(T-s\right)}dW_{s}^{0}.
\end{align*}
Now, we know that $r_{T}$, conditional on $\mathcal{G}_{t}$, is normally
distributed with mean and variance
\begin{align*}
\mathbb{E}\left[r_{T}\mid\mathcal{G}_{t}\right] & =e^{-k\left(T-t\right)}r_{t}+\theta\left(1-e^{-k\left(T-t\right)}\right),\\
\text{Var}\left[r_{T}\mid\mathcal{G}_{t}\right] & =\frac{\sigma^{2}}{2k}\left(1-e^{-2k\left(T-t\right)}\right).
\end{align*}

One can show, see, e.g. \cite{MusielaMarek05} that the price of the zero-coupon bond under the dynamics given in \eqref{eq: Vasicek_Model} is
\begin{align*}
P_{t,T} & = A\left(t,T\right)e^{-B\left(t,T\right)r_{t}},
\end{align*}
where $B\left(t,T\right)\triangleq\frac{1}{k}\left(1-e^{-k\left(T-t\right)}\right)$ and $A\left(t,T\right)\triangleq\exp\left(\left(\theta-\frac{\sigma^{2}}{2k^2}\right)\left(B\left(t,T\right)+t-T\right)-\frac{\sigma^{2}}{4k}B\left(t,T\right)^{2}\right).$
If we now apply Itô's Lemma to $f\left(t,r_{t}\right)=A\left(t,T\right)e^{-B\left(t,T\right)r_{t}}$,
we have
\begin{align*}
dP_{t,T} & =\partial_{t}f\left(t,r_{t}\right)dt+\partial_{r}f\left(t,r_{t}\right)dr_{t}+\frac{1}{2}\partial_{rr}^{2}f\left(t,r_{t}\right)d\left[ r,r\right]_{t}\\
 & =\partial_{t}P_{t,T}dt-A\left(t,T\right)B\left(t,T\right)e^{-B\left(t,T\right)r_{t}}dr_{t}+\frac{1}{2}A\left(t,T\right)B\left(t,T\right)^{2}e^{-B\left(t,T\right)r_{t}}d\left[ r,r\right]_{t}\\
 & =\partial_{t}P_{t,T}dt+P_{t,T}\left(-B\left(t,T\right)dr_{t}+\frac{1}{2}B\left(t,T\right)^{2}d\left[ r,r\right]_{t}\right).
\end{align*}
Replacing the term $dr_{t}$ in the previous equation by its SDE $\left(\ref{eq: Vasicek_Model}\right)$, we have
\begin{align}
\frac{dP_{t,T}}{P_{t,T}} & =-\left(B\left(t,T\right)k\left(\theta-r_{t}\right)-\frac{1}{2}B\left(t,T\right)^{2}\sigma^{2}\right)dt-\sigma B\left(t,T\right)dW_{t}^{0}.\label{eq: Vasicek_ZCbond_Pdynamics}
\end{align}

The forward variance in this case has the following dynamics
\begin{equation}
d\xi_{t,u}=\eta e^{-\kappa\left(u-t\right)}\sqrt{\xi_{t,t}}dW_{t}^{2}.\label{eq: Forward_Variance_SDE_DifferentialForm}
\end{equation}
The Heston model, as any Markovian model, can be rewritten
in forward variance form by means of equations $\left(\ref{eq: Heston_Price}\right)$
and $\left(\ref{eq: Forward_Variance_SDE_DifferentialForm}\right).$ 
The following corollary gives the specific risk-neutral measure for the Vasicek-Heston model, that will be useful for simulation purposes in the next section.
\begin{cor}
\label{cor: Vasicek-Heston_Pricing} 
The risk-neutral measure under the Vasicek-Heston model is given by the measure in \eqref{eq: Q-meas}  with
\begin{eqnarray*}
\gamma_{t}^{0} & =  &\frac{1}{\Theta\left(t,\nu_{t}\right)}
\Big[\eta \sqrt{\nu_{t}}
\Big(2\left(-1+B\left(t,T\right)k\right)r_t + B\left(t,T\right)
\left(B\left(t,T\right)\sigma^{2}-2k\theta\right)\Big)\Big],\\
\gamma_{t}^{1} & = & \frac{-1}{\Theta\left(t,\nu_{t}\right)}
\Big[2B\left(t,T\right)\sigma\eta\left(\mu_t-r_t\right)\Big],\\
\gamma_{t}^{2} & = & \frac{1}{\Theta\left(t,\nu_{t}\right)}
\Big[2B\left(t,T\right)e^{\kappa\left(T-t\right)}r_t\sigma\xi_t\Big],\\
\end{eqnarray*}
where
\[\Theta\left(t,x\right) \triangleq 2B\left(t,T\right)\sigma\eta\sqrt{x}.\]
\end{cor}

\begin{proof}
We will proceed similarly as in Theorem $\ref{thm: Generalized_UL-Pricing_Thm}$. We have
to impose that the discounted price process, $\tilde{S}_{t},$ the
discounted variance swap $\tilde{\xi_{t}},$ and the discounted zero-coupon
bond price $\tilde{P}_{t,T}$ are $\mathbb{Q}-$martingales.
\begin{align*}
d\tilde{S}_{t} & =dB_{t}^{-1}S_{t}+B_{t}^{-1}dS_{t}\\
 & =-r_{t}B_{t}^{-1}S_{t}dt+B_{t}^{-1}\left[\mu_{t}S_{t}dt+S_{t}\sqrt{\nu_{t}}dW_{t}^{1}\right]\\
 & =\tilde{S}_{t}\left(\left[\mu_{t}-r_{t}+\sqrt{\nu_{t}}\gamma_{t}^{1}\right]dt+\sqrt{\nu_{t}}dW_{t}^{\mathbb{Q},1}\right),
\end{align*}
now the discounted price process $\tilde{S}_{t}$ is a $\mathbb{Q}-$martingale
if and only if
\begin{equation}
\gamma_{t}^{1}=\frac{r_{t}-\mu_{t}}{\sqrt{\nu_{t}}}.\label{eq: Price_Q-Martingale_Condition-1}
\end{equation}
We do the same for the discounted forward variance, hence we obtain
\begin{align*}
d\tilde{\xi_{t}} & =dB_{t}^{-1}\xi_{t}+B_{t}^{-1}d\xi_{t}\\
 & =-r_{t}B_{t}^{-1}\xi_{t}dt+B_{t}^{-1}\eta e^{-\kappa\left(T-t\right)}\sqrt{\nu_{t}}dW_{t}^{2}\\
 & =B_{t}^{-1}\left(-r_{t}\xi_{t}dt+\eta e^{-\kappa\left(T-t\right)}\sqrt{\nu_{t}}\left[dW_{t}^{\mathbb{Q},2}+\gamma_{t}^{2}dt\right]\right)\\
 & =B_{t}^{-1}\left(\left[\eta e^{-\kappa\left(T-t\right)}\sqrt{\nu_{t}}\gamma_{t}^{2}-r_{t}\xi_{t}\right]dt+\eta e^{-\kappa\left(T-t\right)}\sqrt{\nu_{t}}dW_{t}^{\mathbb{Q},2}\right),
\end{align*}
therefore the discounted variance swap $\tilde{\xi_{t}}$, is a $\mathbb{Q}-$martingale
if and only if
\begin{equation}
\gamma_{t}^{2} = \frac{\xi_{t}r_{t}}{\eta e^{-\kappa\left(T-t\right)}\sqrt{\nu_{t}}}.\label{eq: FwdVarSwap_Q-Martingale Condition-1}
\end{equation}
Finally, we impose that the discounted zero-coupon bond price process
is a $\mathbb{Q}$-martingale in an analogous computation,
\begin{align*}
d\tilde{P}_{t,T} & =dB_{t}^{-1}P_{t,T}+B_{t}^{-1}dP_{t,T}\\
 & =-r_{t}B_{t}^{-1}P_{t,T}+B_{t}^{-1}P_{t,T}\left[-\left(B\left(t,T\right)k\left(\theta-r_{t}\right)-\frac{1}{2}B\left(t,T\right)^{2}\sigma^{2}\right)dt-\sigma B\left(t,T\right)dW_{t}^{0}\right]\\
 & =-\tilde{P}_{t,T}\left(\left(r_{t}+B\left(t,T\right)k\left(\theta-r_{t}\right)-\frac{1}{2}B\left(t,T\right)^{2}\sigma^{2}\right)dt+\sigma B\left(t,T\right)dW_{t}^{0}\right)\\
 & =-\tilde{P}_{t,T}\left(r_{t}+\sigma B\left(t,T\right)\gamma_{t}^{0}+B\left(t,T\right)k\left(\theta-r_{t}\right)-\frac{1}{2}B\left(t,T\right)^{2}\sigma^{2}\right)dt
-\tilde{P}_{t,T}\left(\sigma B\left(t,T\right)dW_{t}^{\mathbb{Q},0}\right),
\end{align*}
therefore the discounted zero-coupon bond is a $\mathbb{Q}$-martingale
if and only if
\begin{equation}
\frac{-1}{1-B\left(t,T\right)k}\left(\sigma B\left(t,T\right)\gamma_{t}^{0}+B\left(t,T\right)k\theta-\frac{1}{2}B\left(t,T\right)^{2}\sigma^{2}\right)=r_{t}.\label{eq: ZCbond_Q-Martingale_Condition-1}
\end{equation}
The result follows, solving the linear system formed by equations $\left(\ref{eq: Price_Q-Martingale_Condition-1}\right)$,
$\left(\ref{eq: FwdVarSwap_Q-Martingale Condition-1}\right)$ and
$\left(\ref{eq: ZCbond_Q-Martingale_Condition-1}\right)$.
\end{proof}


\section{Model implementation and examples}\label{implementation}

In this section, we present an implementation of the Heston model written
in forward variance together with a Vasicek model for the interest
rates, in order to price numerically a unit-linked product. We will
implement a Monte Carlo scheme for simulating prices under this model
and compare it against a classical Black-Scholes model. The Heston
process will be simulated using a full-truncation scheme \cite{Ander07} in the Euler
discretization on both models. We first show the discretized versions
of the SDE's for each model and the result of the model comparison
given some initial conditions.

Let $N\in\mathbb{N}$ be the number of time steps in which the interval
$\left[0,T\right]$ is equally divided. Then consider the uniform
time grid $t_{k}\triangleq\left(kT\right)/N,$ for all $k=1,\ldots,N$ of length
$\Delta t=T/N.$ We present the following Euler schemes for each model.
\begin{enumerate}
\item Classical Black Scholes
\[
S_{t_{k+1}}=S_{t_{k}}\exp\left(\left(r_0-\frac{1}{2}\bar{\nu}^{2}\right)\Delta t+\bar{\nu}\sqrt{\Delta t}\left(W_{1}^{\mathbb{Q}}\left(t_{k+1}\right)-W_{1}^{\mathbb{Q}}\left(t_{k}\right)\right)\right),
\]
where the parameters for the simulation are $\left(S_{0},r,\bar{\nu}\right),$
given by: $S_{0}=100,$ $r_0=0.01,$ $\bar{\nu}=0.04$.
\item Vasicek-Heston Model written in forward variance 
\begin{align*}
r_{t_{k+1}} & =r_{t_{k}}+\Big[k\Big(\theta-\left(r_{t_{k}}\right)^{+}\Big)+ \sigma\gamma_{0}\left(t_{k}\right)\Big]\Delta t
+\sigma\sqrt{\Delta t}\left(W_{0}^{\mathbb{Q}}\left(t_{k+1}\right)-W_{0}^{\mathbb{Q}}\left(t_{k}\right)\right),\\
\xi_{t_{k+1}}\left(t_{N}\right) & =\xi_{t_{k}}\left(t_{N}\right)+r_{t_{k}}\xi_{t_{k}}\left(t_{N}\right)\Delta t+\eta e^{-\kappa\left(t_{N}-t_{k}\right)}\sqrt{\left(\nu_{t_{k}}\right)^{+}\Delta t}\left(W_{2}^{\mathbb{Q}}\left(t_{k+1}\right)-W_{2}^{\mathbb{Q}}\left(t_{k}\right)\right),\\
\nu_{t_{k+1}} & =\bar{\nu} + e^{\kappa\left(t_{N}-t_{k+1}\right)}\left(\xi_{t_{k+1}}\left(t_{N}\right)-\bar{\nu}\right),\\
S_{t_{k+1}} & = S_{t_{k}}+ r_{t_{k}}S_{t_{k}}\Delta t
+S_{t_{k}}\sqrt{\left(\nu_{t_{k}}\right)^{+}\Delta t}\left(W_{1}^{\mathbb{Q}}\left(t_{k+1}\right)-W_{1}^{\mathbb{Q}}\left(t_{k}\right)\right),\\
P\left(t_{k+1},t_{N}\right) & =P\left(t_{k},t_{N}\right)+P\left(t_{k},t_{N}\right)\left[r_{t_{k}}\Delta t-\sigma B\left(t_{k},t_{N}\right)\sqrt{\Delta t}\left(W_{0}^{\mathbb{Q}}\left(t_{k+1}\right)-W_{0}^{\mathbb{Q}}\left(t_{k}\right)\right)\right].
\end{align*}where the parameters are $\left(S_{0},\mu,\nu_{0},\bar{\nu},\kappa,\eta,r_{0},\theta,k,\sigma\right)$
and were set as $S_{0}=100,$ $\mu=0.015,$ $\bar{\nu}=0.01,$ $\nu_{0}=0.04,$ $\kappa=10^{-3},$ $\eta=0.01,$ 
$\theta=r_{0}=0.01,$ $k=0.3,$ and $\sigma=0.02$.
\end{enumerate}
For simulation purposes, the Monte Carlo scheme was implemented using
5,000 simulations. The following graphs in Figure \ref{plot: model_errors}, results from the implementation
of the previous models with the mentioned initial conditions, and
for $T=\left\{ 10,20,30,40\right\} .$

\begin{figure}
\begin{centering}
\includegraphics[width=0.49\textwidth]{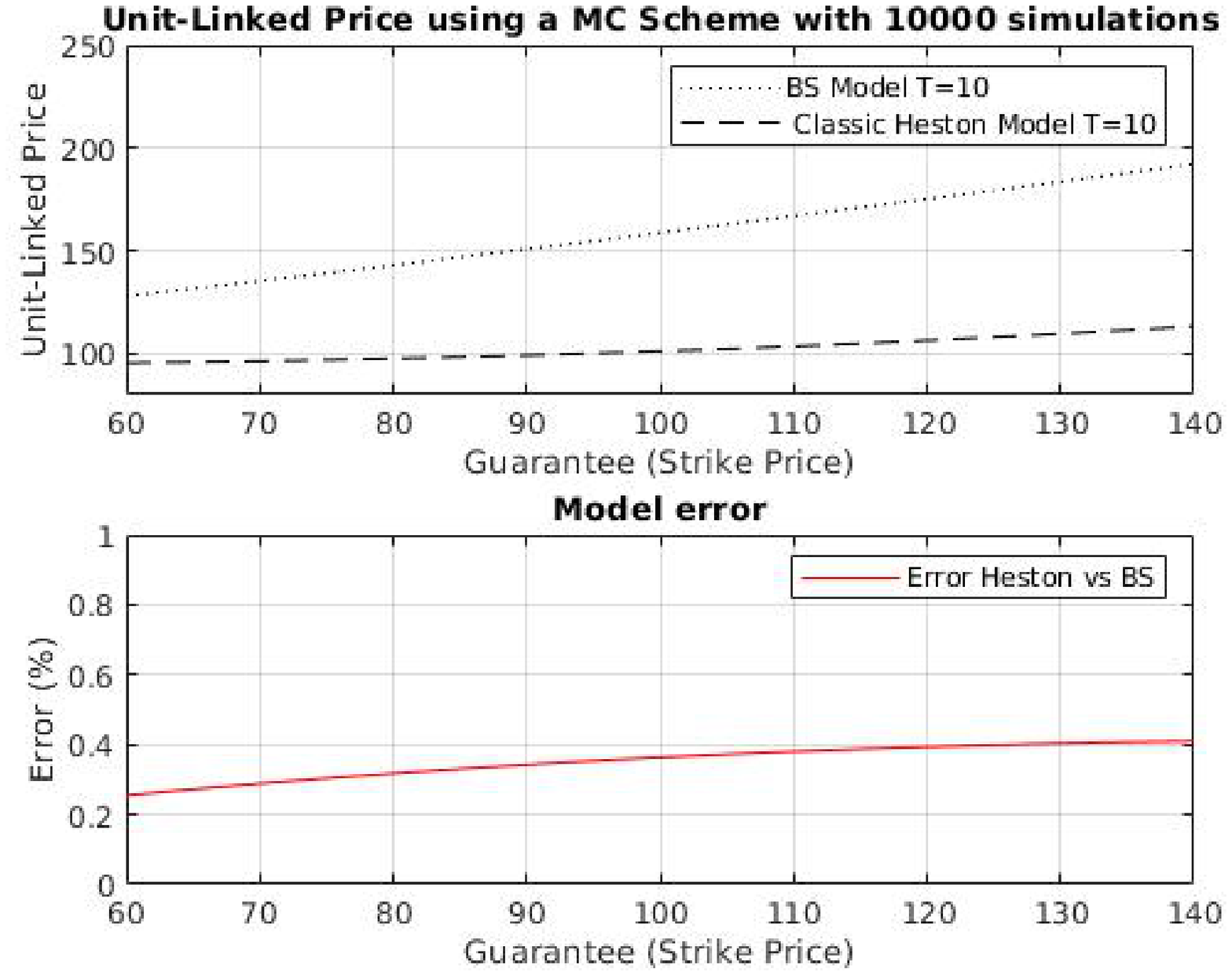}\includegraphics[width=0.49\textwidth]{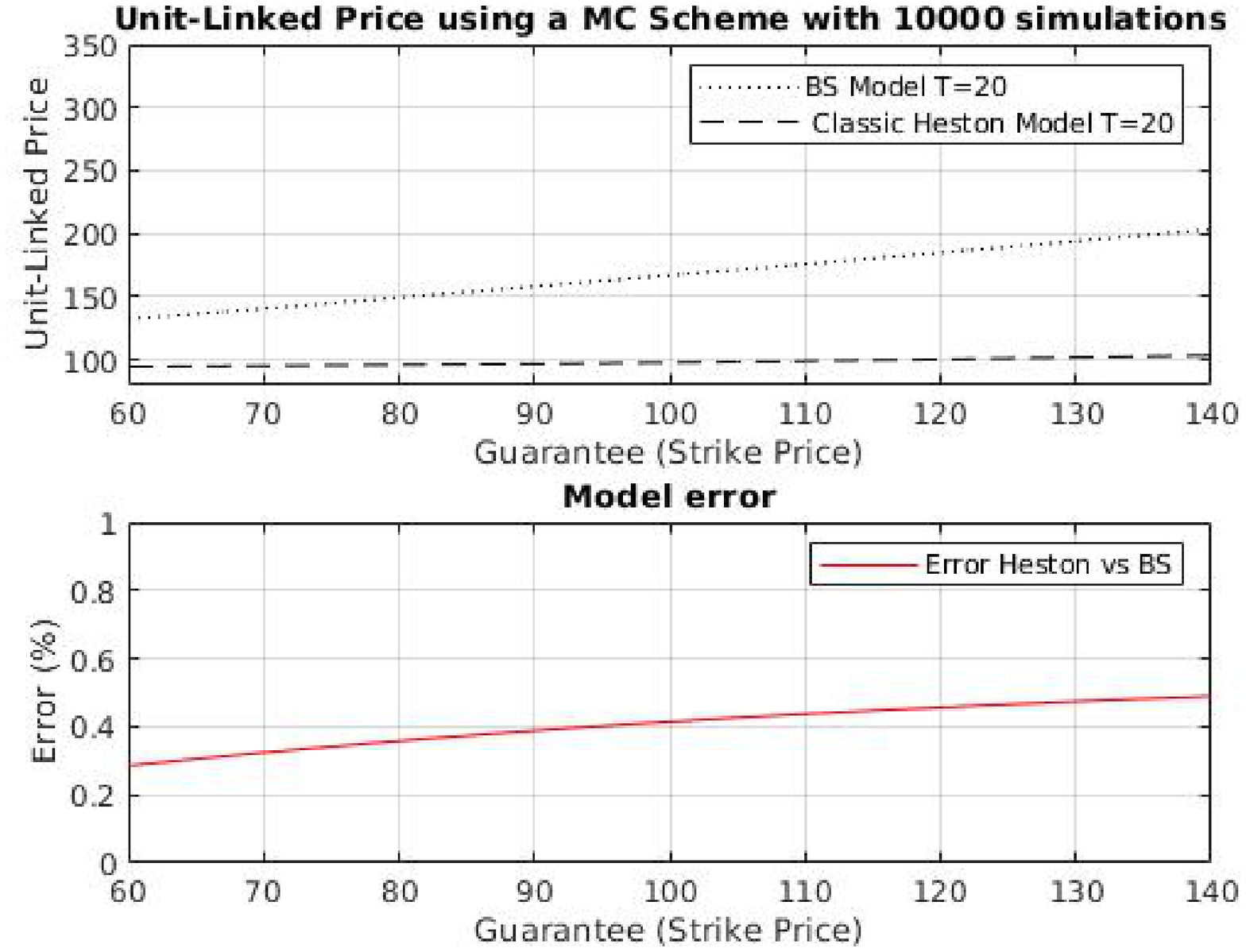}
\par\end{centering}
\includegraphics[width=0.49\textwidth]{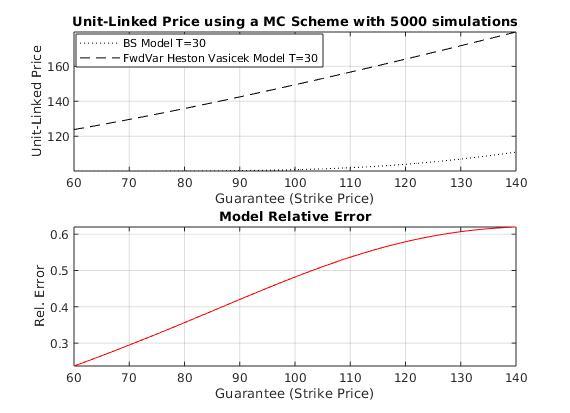}\includegraphics[width=0.49\textwidth]{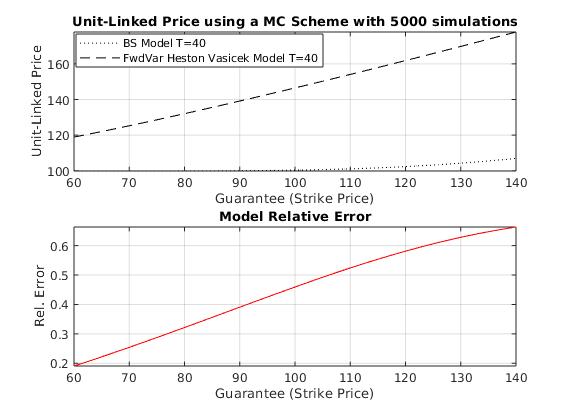}

\caption[Model error]{Pricing comparative between the Black-Scholes model and a Vasicek-Heston
model written in forward variance with the mentioned initial conditions
and $T\in\left\{10,20,30,40\right\} $.}
\label{plot: model_errors}
\end{figure}
\newpage{}

As shown in Theorem $\ref{thm: Generalized_UL-Pricing_Thm}$ and Corollary
$\ref{cor: Vasicek-Heston_Pricing}$, in order to properly price a
unit-linked product, it only remains to multiply the value of the
derivative priced using the Monte Carlo scheme, times the probability that an
$x$-year old insured survives during the life of the product ($T$
years). To do so, we have used Norwegian mortality from 2018 extracted from Statistics Norway.

\begin{table}[H]
\begin{center}
\begin{tabular}{ l | l l | l}
 Age & Men & Women & Total\\ \hline 
 4 & 50 & 45 & 95\\
 9 & 7 & 2 & 9\\
 14 & 10 & 3 & 13\\
 19 & 26 & 13 & 39\\
 24 & 33 & 6 & 39\\
 29 & 63 & 24 & 87\\
 34 & 72 & 27 & 99\\
 39 & 93 & 43 & 136\\
 44 & 109 & 68 & 177\\
 49 & 156 & 111 & 267\\
 54 & 258 & 177 &435 \\
 59 & 454 & 310 & 764\\
 64 & 737 & 495 & 1232\\
 69 & 1206 & 824 & 2030\\
 74 & 1990 & 1331 & 3321\\
 79 & 3602 & 2447 & 6049\\
 84 & 6626 & 4628 & 11254\\
 89 & 12469 & 9053 & 21522\\
 $\geq $90 & 21 909 & 24 230 & 46139
\end{tabular}
\end{center}
\caption{Norwegian mortality in 2018, per 100 000 inhabitants. Data from Statistics Norway, table: 05381}
\label{mortalityNorway}
\end{table}
As it is usual, mortality among men is higher, we consider however, the aggregated mortality for simplicity. To model the mortality given in Table \ref{mortalityNorway} we use the Gompertz-Makeham law of mortality which states that the death rate is the sum of an age-dependent component, which increases exponentially with age, and an age-independent component, i.e. $\mu_{\ast\dag}(t) = a+be^{ct}$, $t\in\left[0,T\right]$. This law of mortality describes the age dynamics of human mortality rather accurately in the age window from about 30 to 80 years of age, which is good enough for our purposes. For this reason, we excluded the very first and last observations from the table. We then find the best fit for $\mu_{\ast\dag}$ in the class of functions $\mathcal{C}=\{f(t)=a+be^{ct}, t\in\left[0,T\right], a,b,c\in \mathbb{R}\}$. As stated previously, since the stochastic process $X=\left\{ X_{t}\right\} _{t\in\left[0,T\right]}$, which regulates the states of the insured, is a regular Markov chain, then the survival probability of an $x$-year old individual during
the next $T$ years is
\[
_{T}p_{x}=\bar{p}_{\ast\ast}\left(x,x+T\right)=\exp\left(-\int_{x}^{x+T}\mu_{\ast\dagger}\left(\tau\right)d\tau\right).
\]

Figure \ref{Figure: MortalityRates} shows the fitted Gompertz-Makeham law based on the mortality data from Table \ref{mortalityNorway}.

\begin{figure}
\begin{center}
\includegraphics[width=0.8\textwidth]{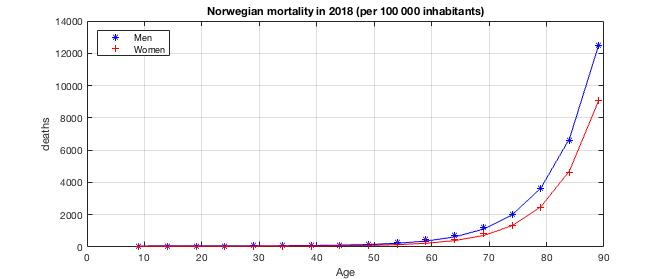}
\end{center}
\caption[Mortality Graph]{Joint plot of the mortality data given in Table \ref{mortalityNorway}, together with the fitted curve using the Gompertz-Makeham law of mortality.}
\label{Figure: MortalityRates}
\end{figure}

Now, using the Vasicek-Heston model written in forward variance, we
can compute a unit-linked price surface in terms of the guarantee,
or strike price, and the age of the insured given a terminal time
for the product $T>0$. In particular, the graphs below show the price
surfaces for fixed $T=\left\{ 10,20,30,40\right\} $. 

\begin{figure}
\begin{centering}
\includegraphics[width=0.49\textwidth]{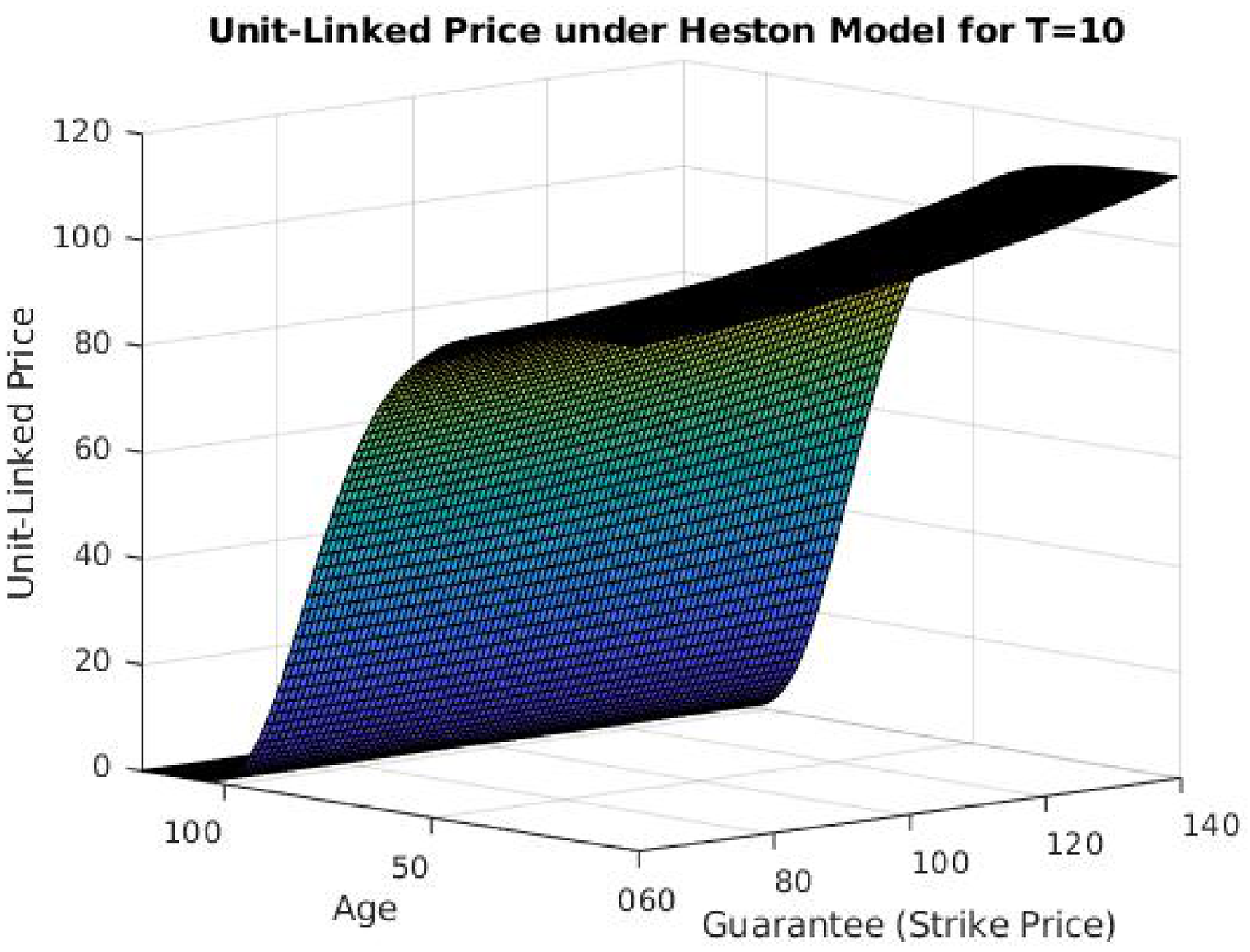}\includegraphics[width=0.49\textwidth]{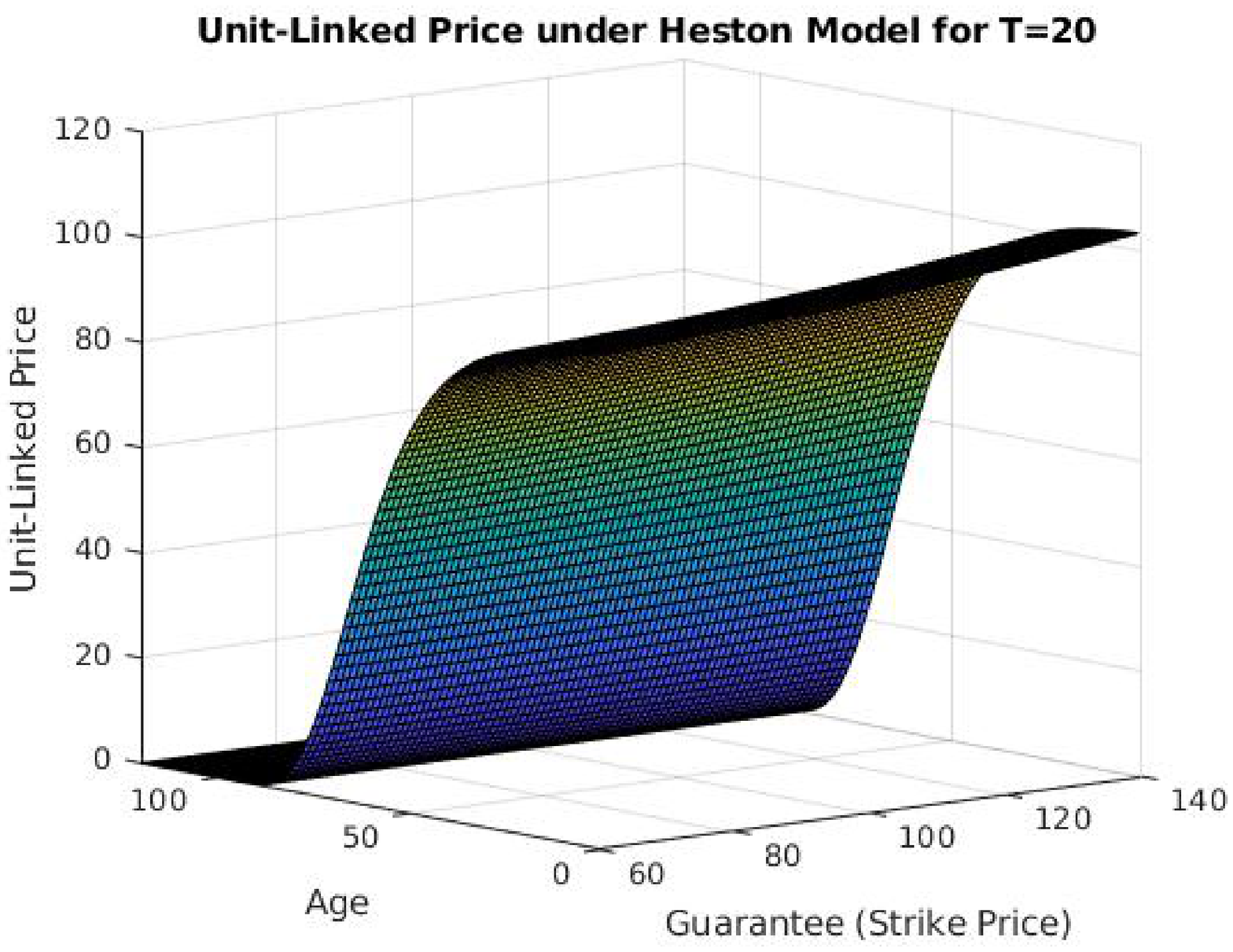}
\par\end{centering}
\includegraphics[width=0.49\textwidth]{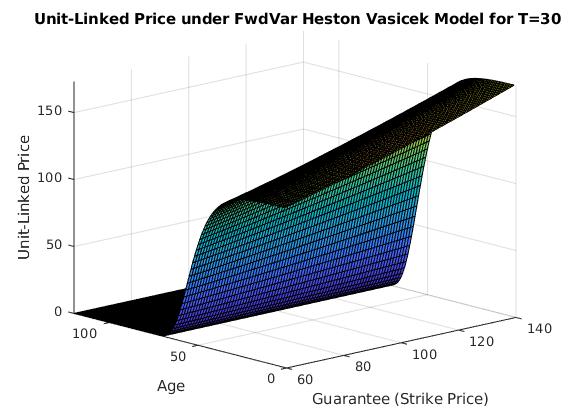}\includegraphics[width=0.49\textwidth]{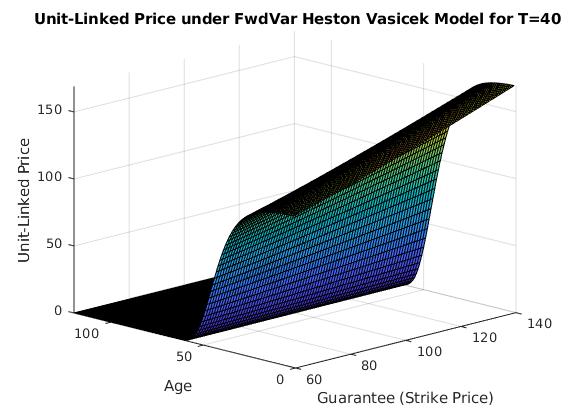}

\caption[Price surfaces]{unit-linked price surfaces under a Vasicek-Heston model written in
forward variance for different policy maturities, in terms of the
guaranteed amount desired by the insured and his age at the time of
acquisition.}
\label{plot: price surfaces}
\end{figure}
From the plots in Figure \ref{plot: price surfaces}, we can observe that the longer time to maturity
is, the lower the unit-linked price is, since the less probable it is that the
insured survives. This effect has greater impact in the price, than
the effect of future volatility, or uncertainty arising from the stochasticity in interest rates. 
This behavior is easily observed by noting how the price surface collapses to zero
as the contract's time to maturity increases, as well as the age of
the insured when entering the contract. Hence, we can say that time to maturity has a cancelling effect on
price, i.e. on one hand it increases price as the stock or fund pays
longer performance, but on the other hand it decreases price due to
a lower probability of surviving during the time to maturity of the
unit-linked contract.

The following plots in Figures \ref{Figure: Price_Histogram} and \ref{Figure: QQPlots}, are aimed at providing the reader with an overview of the distributional properties of the price process at a constant survival rate equal to one. The first thing that comes to sight, is how the variance and time to maturity are directly proportional. Also, the longer the time to maturity of the unit-linked product is, the more leptokurtic the distribution of the insurance product price is. This is an important thing to take into account in the modeling of prices due to the impact in the hedging of such insurance products. 

\begin{figure}[H]
\begin{center}
\includegraphics[width=0.8\textwidth]{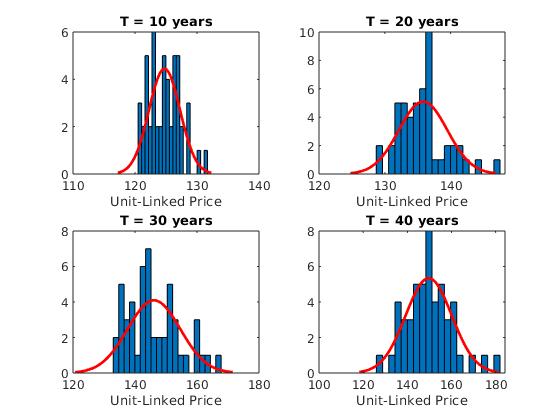}
\end{center}
\caption[UL-Histogram]{unit-linked price histograms with constant survival rate equal to 1.}
\label{Figure: Price_Histogram}
\end{figure}
\begin{figure}[H]
\begin{center}
\includegraphics[width=0.8\textwidth]{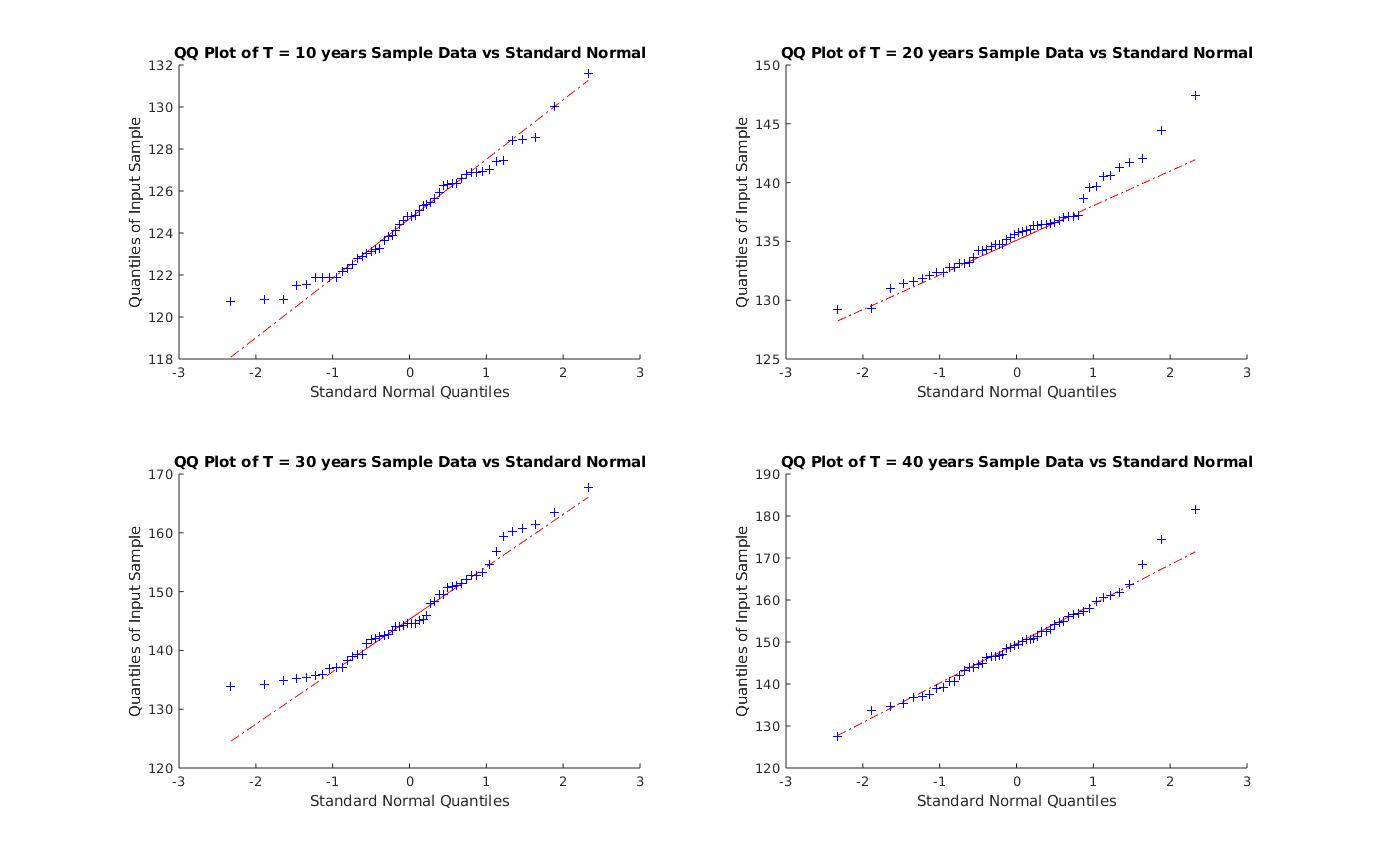}
\end{center}
\caption[QQ-Plot]{QQ-Plot between the unit-linked price input data and the Standard Normal Distribution for maturities $T=\{10,20,30,40\}$ years.}
\label{Figure: QQPlots}
\end{figure}

\subsection{Pure endowment}
Consider an endowment for a life aged $x$ with maturity $T>0$. The policy pays the amount $E_T:=\max\{G_e, S_T\}$ if the insured survives by time $T$ where $G_e>0$ is a guaranteed amount and $S_T$ is the value of a fund at the expiration time. This policy is entirely determined by the policy function
$$a_\ast(t) = \begin{cases} E_T \mbox { if } t\geq T\\ 0 \mbox{ else}\end{cases}.$$

In view of \eqref{reserve2} and the above function, the value of this insurance at time $t$ given that the insured is still alive is then given by
\begin{align}\label{reserve3}
V_\ast^+(t,A)= \mathbb{E}^{\mathbb{Q}}\left[\int_t^T \frac{B_t}{B_s}p_{\ast\ast}(x+t,x+s) da_\ast(s)\Big|\mathcal{G}_t\right]=\mathbb{E}^{\mathbb{Q}}\left[\frac{B_t}{B_T} E_T\Big|\mathcal{G}_t\right] p_{\ast\ast}(x+t,x+T),
\end{align}
The above quantity corresponds to the formula in Theorem \ref{thm: Generalized_UL-Pricing_Thm}.

Observe that, the payoff of an endowment can be written as
$$\max\{G_e,S_T\} = (G_e-S_T)_+ + G_e,$$
where $(x)_+ \triangleq \max\{x,0\}$, which corresponds to a call option with strike price $G_e$ plus $G_e$. In the case that $S$ is modelled by the Black-Scholes model (with constant interest rate) we know that the price at time $t$ of a call option with strike $G_e$ and maturity $T$ is given by
$$BS(t,T,S_t,G_e)\triangleq\Phi (d_1(t,T))S_t -\Phi (d_2(t,T))G_e e^{-r(T-t)} ,$$
where $\Phi$ denotes the distribution function of a standard normally distributed random variable and
$$d_1(t,T) \triangleq \frac{\log (S_t/G_e) +\left(r+\frac{1}{2}\sigma^2\right)(T-t)}{\sigma\sqrt{T-t}}, \quad d_2(t,T) \triangleq d_1(t,T) - \sigma\sqrt{T-t}.$$

Then we have that the unit-linked pure endowment under the Black-Scholes model has the price
\begin{align}\label{BSpure}
BSE(t,T,S_t,G_e)\triangleq \Phi(d_1(t,T))S_t + G_e e^{-r(T-t)} \Phi(-d_2(t,T)).
\end{align}

%

The single premium at the beginning of this contract under the Black-Scholes model is then
$$\pi_{BS}^0 \triangleq BSE(0,T,S_0,G_e).$$

It is also possible to compute yearly premiums by introducing payment of yearly premiums $\pi_{BS}$ in the policy function $a_{\ast}$, i.e. $a_\ast(t)=-\pi_{BS} t$ if $t\in [0,T)$ and $a_\ast(t) = -\pi_{BS} T+E_T$ if $t\geq T$, then the value of the insurance at any given time $t\geq 0$ with yearly premiums, denoted by $V_\ast^\pi$, becomes
$$-\pi_{BS}\int_t^T e^{-r(s-t)}p_{\ast\ast}(x+t,x+s)ds+BSE(t,T,S_t,G_e).$$
We choose the premiums in accordance with the equivalence principle, i.e. such that the value today is $0$,
$$\pi_{BS}= \frac{BSE(0,T,S_0,G_e)}{\int_0^T e^{-rs}p_{\ast\ast}(x,x+s)ds}.$$

Under the Vasicek-Heston model instead, the value of policy at time $t\geq 0$ with yearly premiums $\pi_{VH}$ is
\begin{align*}
V_\ast^+(t,A) &= \mathbb{E}^{\mathbb{Q}}\left[\int_t^T \frac{B_t}{B_s}p_{\ast\ast}(x+t,x+s) da_\ast(s)\Big|\mathcal{G}_t\right]\\
&=-\pi_{VH}\int_t^T \mathbb{E}^{\mathbb{Q}}\left[\frac{B_t}{B_s}\Big|\mathcal{G}_t\right]p_{\ast\ast}(x+t,x+s)ds+ \mathbb{E}^{\mathbb{Q}}\left[\frac{B_t}{B_T} E_T\Big|\mathcal{G}_t\right] p_{\ast\ast}(x+t,x+T).
\end{align*}

A single premium payment $\pi_{VH}^0$ corresponds to $V_\ast^+(0,A)$, i.e.
$$\pi_{VH}^0 = \mathbb{E}^{\mathbb{Q}}\left[\frac{E_T}{B_T}\right] p_{\ast\ast}(x,x+T)$$
and the yearly ones correspond to
$$\pi_{VH} = \frac{V_\ast^+(0,A)}{\int_0^T \mathbb{E}^{\mathbb{Q}}\left[\frac{1}{B_s}\right]p_{\ast\ast}(x,x+s)ds}.$$

In Figure \ref{figure:pure_endowment} we compare the single premiums using the classical Black-Scholes unit-linked model in contrast to the Vasicek-Heston model proposed for different maturities $T$ with parameters $S_0=1$, $G_e=1$, $r=1\%$ and $\mu=1.5\%$, $\sigma=4\%$ for the Black-Scholes model, and $S_{0}=1,$, $G_e=1$, $\mu=1.5\%,$ $\bar{\nu}=1\%,$ $\nu_{0}=4\%,$ $\kappa=10^{-3},$ $\eta=10^{-2},$ 
$\theta=r_{0}=1\%,$ $k=0.3,$ $\sigma=2\%$ for the Vasicek-Heston model.
\begin{figure}
\centering
  \includegraphics[scale=0.5]{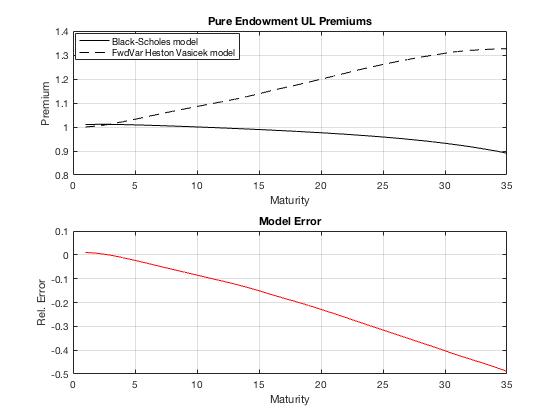}
\caption{Single premiums for a pure endowment with benefit equal to $1$ monetary unit, using the classical Black-Scholes with constant interest $r=1\%$ and $\mu=1.5\%$, $\sigma=4\%$ and a Vasicek-Heston model with parameters $S_{0}=1,$ $\mu=1.5\%,$ $\bar{\nu}=1\%,$ $\nu_{0}=4\%,$ $\kappa=10^{-3},$ $\eta=10^{-2},$ 
$\theta=r_{0}=1\%,$ $k=0.3,$ $\sigma=2\%$.}\label{figure:pure_endowment}
\end{figure}

\subsection{Endowment with death benefit}
Consider now an endowment for a life aged $x$ with maturity $T>0$ that pays, in addition, a death benefit in case the insured dies within the period of the contract. That is the policy pays the amount $E_T:=\max\{G_e, S_T\}$ if the insured survives by time $T$ as before and, in addition, a death benefit of $D_t :=\max\{G_d, S_t\}$ if $t\in [0,T)$. This policy is entirely determined by the two policy functions
$$a_\ast(t) = \begin{cases} E_T \mbox { if } t\geq T\\ 0 \mbox{ else}\end{cases},\quad  a_{\ast\dag}(t) = \begin{cases} D_t \mbox { if } t\in [0,T)\\ 0 \mbox{ else}\end{cases}.$$

In view of \eqref{reserve2} and the above functions, the value of this insurance at time $t$ given that the insured is still alive is then given by
\begin{align}\label{reserve4}
V_\ast^+(t,A)= \mathbb{E}^{\mathbb{Q}}\left[\frac{B_t}{B_T} E_T\Big|\mathcal{G}_t\right] p_{\ast\ast}(x+t,x+T) + \int_t^T \mathbb{E}^{\mathbb{Q}}\left[\frac{B_t}{B_s} D_s\Big| \mathcal{G}_t\right] p_{\ast\ast}(x+t,x+s)\mu_{\ast\dag}(x+s) ds.
\end{align}

Following similar arguments as in the case of a pure endowment, by adding the function $a_{\ast\dag}$ in the computations, we obtain that the single premiums $\pi_{BS}^0$ and $\pi_{VH}^0$ for the Black-Scholes model and Vasicek-Heston model, respectively, are given by.

$$\pi_{BS}^0 = BSE(0,T,S_0,G_e) + \int_0^T e^{-r s} BSE(0,s,S_0,G_d) p_{\ast\ast}(x,x+s)\mu_{\ast\dag}(x+s)ds,$$
where the function $BSE$ is given in \eqref{BSpure}, and
$$\pi_{VH}^0 = \mathbb{E}^{\mathbb{Q}}\left[\frac{E_T}{B_T}\right] p_{\ast\ast}(x,x+T) + \int_0^T \mathbb{E}^{\mathbb{Q}}\left[\frac{D_s}{B_s}\right] p_{\ast\ast}(x,x+s)\mu_{\ast\dag}(x+s) ds..$$

In Figure \ref{figure:death_benefit} we compare the single premiums using the classical Black-Scholes unit-linked model in contrast to the Vasicek-Heston model proposed for different maturities $T$ with parameters $S_0=1$, $G_e=G_d=1$, $r=1\%$ and $\mu=1.5\%$, $\sigma=4\%$ for the Black-Scholes model, and $S_{0}=1,$, $G_e=1$, $\mu=1.5\%,$ $\bar{\nu}=1\%,$ $\nu_{0}=4\%,$ $\kappa=10^{-3},$ $\eta=10^{-2},$ 
$\theta=r_{0}=1\%,$ $k=0.3,$ $\sigma=2\%$ for the Vasicek-Heston model.

\begin{figure}[H]
\centering
  \includegraphics[scale=0.5]{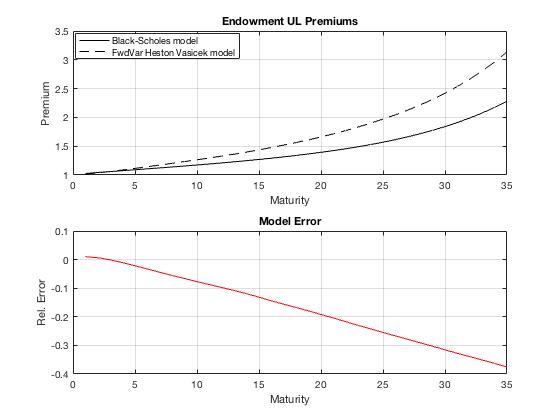}
\caption{Single premiums for an endowment with benefits equal to $1$ monetary unit, using the classical Black-Scholes with constant interest $r=1\%$ and $\mu=1.5\%$, $\sigma=4\%$ and a Vasicek-Heston model with parameters $S_{0}=1,$ $\mu=1.5\%,$ $\bar{\nu}=1\%,$ $\nu_{0}=4\%,$ $\kappa=10^{-3},$ $\eta=10^{-2},$ 
$\theta=r_{0}=1\%,$ $k=0.3,$ $\sigma=2\%$.}\label{figure:death_benefit}
\end{figure}

\bibliography{BibBLO2020}

\begin{thebibliography}{10}

\bibitem{AasePersson2011}
K.~K. Aase and S.-A. Persson.
\newblock Pricing of unit-linked life insurance policies.
\newblock {\em Scandinavian Actuarial Journal}, 1994(1):26--52, 1994.

\bibitem{Aly14}
S.M.~Ould Aly.
\newblock Forward variance dynamics: {B}ergomi's model revisited.
\newblock {\em Appl. Math. Finance}, 21(1):84--107, 2014.

\bibitem{Ander07}
L.~Andersen.
\newblock Efficient simulation of the heston stochastic volatility model.
\newblock {\em J. Computat. Finance}, 11, 01 2007.

\bibitem{BO93}
A.~R. Bacinello and S.-A. Persson.
\newblock Pricing guaranteed securities-linked life insurance under interest
  rate risk.
\newblock {\em Proceedings from AFIR (Actuarial Approach for Financial Risks)
  3rd International Colloquium, Rome, Italy}, 1:35--55, 1993.

\bibitem{BP02}
A.~R. Bacinello and S.-A. Persson.
\newblock Design and pricing of equity-linked life insurance under stochastic
  interest rates.
\newblock {\em Journal of Risk Finance}, 3(2):6--21, 2002.

\bibitem{BergomiGuyon2012}
L.~Bergomi and J.~Guyon.
\newblock Stochastic volatility's orderly smiles.
\newblock {\em Risk}, 25(5):60--66, 2012.

\bibitem{BoyleSchwartz77}
P.~Boyle and E.~Schwartz.
\newblock Equilibrium prices of guarantees under equity-linked contracts.
\newblock {\em The Journal of Risk and Insurance}, 44(4):639--660, 1977.

\bibitem{Fili09}
D.~Filipovi\'c.
\newblock {\em Term-Structure Models}.
\newblock Springer Finance. Springer, Berlin, Heidelberg, 2009.
\newblock A Graduate Course.

\bibitem{Koller}
M.~Koller.
\newblock {\em Stochastic Models in Life Insurance}.
\newblock Springer, second edition, 2012.

\bibitem{MusielaMarek05}
Marek Musiela and Marek Rutkowski.
\newblock {\em Martingale methods in financial modelling}, volume~36 of {\em
  Stochastic Modelling and Applied Probability}.
\newblock Springer-Verlag, Berlin, second edition, 2005.

\bibitem{Revuz99}
D.~Revuz and M.~Yor.
\newblock {\em Continuous martingales and {B}rownian motion}, volume 293 of
  {\em Grundlehren der Mathematischen Wissenschaften [Fundamental Principles of
  Mathematical Sciences]}.
\newblock Springer-Verlag, Berlin, third edition, 1999.

\bibitem{RomanoTouzi97}
M.~Romano and N.~Touzi.
\newblock Contingent claims and market completeness in a stochastic volatility
  model.
\newblock {\em Math. Finance}, 7(4):399--412, 1997.

\bibitem{WaQiWa13}
W.~Wang, L.~Qian, and W.~Wang.
\newblock Hedging strategy for unit-linked life insurance contracts in
  stochastic volatility models.
\newblock {\em WSEAS}, 12(4):363--373, 2013.

\end{thebibliography}

\bibliographystyle{plain}

\end{document}